\documentclass[final,12pt]{article}
\usepackage{amsfonts}
\usepackage{amssymb}
\usepackage{amsmath}
\usepackage{geometry}
\usepackage{fancyhdr}
\usepackage[bottom]{footmisc}
\usepackage{version}
\usepackage{hyperref}
\usepackage{indentfirst}
\usepackage{natbib}

\newtheorem{definition}{Definition}
\newtheorem{Assumption}{Assumption}

\newtheorem{lemma}{Lemma}

\newtheorem{proposition}{Proposition}

\newenvironment{proof}[1][Proof]{\noindent \textit{#1.} }{}
\begin{document}

\begin{center}
{\large Testable Implications of Multiple Equilibria in \\ Discrete Games with
Correlated Types}\footnote{de Paula would like to gratefully acknowledge financial support
from the Economic
and Social Research Council (ESRC) under ESRC Centre for Microdata Methods and Practice (CeMMAP), grant number RES589-28-0001. We thank Kristof Kutasi for excellent research assistance. 
Correspondence: \texttt{xt9@rice.edu, a.paula@ucl.ac.uk}.} \bigskip 

\'{A}ureo de Paula

University College London, CeMMAP and IFS\bigskip

Xun Tang

Rice University\bigskip

\today%
\bigskip

Abstract\bigskip
\end{center}

We study testable implications of multiple equilibria in discrete games  with incomplete information. 
Unlike \cite{dePaulaTang2012}, we allow the players' private signals to be correlated. 
In static games, we leverage independence of private types across games whose equilibrium selection is correlated. 
In dynamic games with serially correlated discrete unobserved heterogeneity, our testable implication builds on the fact that the distribution of a sequence of choices and states are mixtures over equilibria and unobserved heterogeneity. 
The number of mixture components is a known function of the length of the sequence as well as the cardinality of equilibria and unobserved heterogeneity support. 
In both static and dynamic cases, these testable implications are implementable using existing statistical tools. \pagebreak

\section{Introduction}

In many social-economic contexts, individuals or institutions interact strategically in response to private incentives. Empirical studies typically model such interaction via static or dynamic games with incomplete information, and exploit the equilibrium implications to infer the incentives from observed states and actions. Examples include location choices in the video retail industry in \cite{seim2006}, timing of commercials by radio stations in \cite{sweeting2009}, choices of effort by students and teachers in classrooms in \cite{toddWolpin2018}, market entry and exit of grocery stores in \cite{grieco2014}, and the dynamic demand and supply in shipbuilding industry in \cite{kalouptsidi2018}. Popular methods for estimating these games require a reduced-form first step that estimates the conditional choice probabilities (CCPs) in equilibrium. Subsequent steps use model-implied links between these CCPs and structural parameters to infer the latter. See \cite{aradillas2010}, \cite{aradillas2019}, and \cite{bajari_et_al2010} for static games; and see \cite{aguirregabiriaMira2007}, \cite{pesendorferDengler2008} and \cite{arcidiaconoMiller2011} for dynamic games.

Given a set of model parameters, games with private types generally admit multiple equilibria. While multiplicity does not necessarily preclude model identification\footnote{As \cite{sweeting2009} notes, multiple equilibria in a sample can sometimes be used to aid the identification and estimation of players' payoffs in parametric likelihood models.} and estimation protocols exist for set identified models, multiple equilibria in a data-generating process (DGP) pose challenges to estimation in both static and dynamic settings. A researcher using observational data typically has no knowledge a priori as to whether a sample is generated from a single equilibrium. With multiple equilibria in the data, the reduced-form estimator in the first step in CCP-based methods, for example, converges to a mixture of CCPs from each equilibrium represented in the sample. This mixture does not conform to the structural links used in sequential estimation.\footnote{Formally, a mixture of several CCP vectors, each indexed by a different equilibrium, does not generally satisfy structural equations that characterise a single equilibrium.} Therefore, detecting multiple equilibria in the sample can be an essential step for valid inference of player incentives in these settings.

In the first part of this paper, we propose a way to test multiple Bayesian Nash equilibria (BNE) in static Bayesian games where players' private types are correlated, for example, through game-level unobserved heterogeneity that is known to all players but not reported in the data. \cite{dePaulaTang2012} proposed a simple test for multiple equilibria in static Bayesian games where the private types are independent conditional on observed states. That test exploited a simple equilibrium implication. With a single BNE in data, players' strategic choices are independent conditional on covariates when we pool observations in the sample. On the other hand, with multiple equilibria in data, those choices are correlated due to their co-movement across different equilibria in the sample. The sign of such correlation conforms with the sign of interaction effects. This approach in \cite{dePaulaTang2012} does not apply when private types are correlated conditioning on observed covariates. In this case, strategic choices are not independent even when there is a single equilibrium.

We introduce a new method to detect multiple equilibria in static Bayesian games. This requires a working assumption that an empiricist can group the games in a sample into pairs or clusters within which equilibrium selection, if there are multiple equilibria, is correlated. For example, consider married couples making joint retirement decisions in a context of simultaneous Bayesian games. Our method can be applied if households with similar demographics and social-economic status, or located in the same geographic area, tend to be correlated in how they adopt strategies from multiple equilibria. The private types are drawn independently across games in the same pair or cluster, but are correlated within each game. Thus, strategic choices of players from two different games in the same cluster are independent under the null of single equilibrium, but are generically correlated under the alternative due to correlated equilibrium selection. Intuitively, permuting players across games within the same ``cluster'' allows one to emulate the ideas in \cite{dePaulaTang2012}. This implication lends itself to a simple test for multiple equilibria, which is analogous to inference of covariate relevance in nonparametric regressions.

\cite{aguirregabiriaMira2019} provided identification results in static Bayesian games with multiple equilibria and discrete unobserved heterogeneity. Their strategy requires the private types be independent from discrete unobserved heterogeneity and observed covariates, and rank conditions on equilibrium CCPs.\footnote{See Assumption 1(B) and condition (d) in Proposition 1 in \cite{aguirregabiriaMira2019}.} While identifying the full model, \cite{aguirregabiriaMira2019} also introduced a novel idea for detecting multiple equilibria. The idea is to compare the ex post payoffs of players recovered from each profile of CCPs as indexed by unobserved heterogeneity and equilibria. To implement this, one needs to use exclusion restrictions in the ex post payoffs, and know the actual distribution of private types.

Our approach differs from that idea in \cite{aguirregabiriaMira2019} in several ways. First, we separate the task of detecting multiple equilibria from that of identifying the full model. Thus, we do not require conditions for identifying the full model, such as exclusion restrictions in ex post payoffs, rank conditions on CCPs, and assumed knowledge of private type distribution. Second, we allow the correlation between player choices to arise from continuous game-level unobserved heterogeneity that can be correlated with other idiosyncratic errors as well. Lastly, inference of single equilibrium using our approach is analogous to tests of covariate relevance in nonparametric regressions, and can be implemented using existing procedures such as \cite{racine_et_al2006}. It does not require sequential steps of estimating and then comparing ex post payoffs from different profiles of CCPs. We elaborate on these substantial differences in Section 2.4.

In the second part of the paper, we propose a new test for multiple Markov Perfect Equilibria (MPE) in dynamic games where private types are correlated due to Markovian, game-level discrete unobserved heterogeneity (DUH) in each period.\footnote{A Markov Perfect Equilibrium is a profile of time-homogeneous pure strategies that map a player's information in each single time period to a choice. Following convention in the literature, we maintain that players do not switch between equilibria within the process of a dynamic game.} In this case, a player's single-period payoff depends on a time-varying DUH, through which players' private information is correlated. If the DUH is drawn independently in each period with no serial correlation, then the idea we introduce for static games above is applicable, with each dynamic game in the sample serving as a ``cluster'' itself. That is, we can test the hypothesis of multiple MPE by checking whether the choices made by two players in two different time periods are correlated. However, this does not work with a serially correlated DUH, which causes correlation between those choices even under the null of single equilibrium.

To meet this challenge, we propose a new way to identify the cardinality of equilibria in dynamic games with private types and serially correlated DUH. Our method exploits the feedbacks of earlier states and choices on future outcomes over different time horizons in an equilibrium Markov process. The idea is to extract information about equilibrium selection and DUH from how these feedbacks vary with the length of history examined. We partition and pair the history of outcomes so that their joint distribution takes the form of a finite mixture. The number of components in the finite mixture is identifiable by existing methods such as \cite{kasaharaShimotsu2014}. More importantly, we show how the number of components in this mixture is determined by the length of history $T$ as well as the cardinality of equilibria and the support of DUH.

To the best of our knowledge, our method uses a new source of variation that has not been exploited in the literature for detecting multiple equilibria in dynamic games. Using these results, one can conduct inference for the cardinality of equilibria in a sample, using standard rank tests such as \cite{kleibergenPaap2006}. This dispenses with any need for identifying a full MPE model and comparing estimates of structural elements, such as an alternative, sequential approach mentioned in \cite{luo_et_al2019}.

The rest of the paper is organized as follows. Section 2 presents the testable implications in static Bayesian games, and discusses how a test can be constructed as a test of covariate relevance in nonparametric regressions. We also illustrate the testable implications and the necessary assumptions using numerical examples. Section 3 does the same for dynamic games with private types, where the private information is correlated through serially correlated DUH.  We further discuss the related literature at the end of each section.

\bigskip

\noindent \textbf{Notation}. We use uppercase letters to denote random
variables and vectors, and use lowercase letters to denote their realized
values. We use calligraphic letters such as $\mathcal{X}$ for the support of
random variables or vectors, and let $\#\mathcal{X}$ denote its cardinality.
For a generic random vector $R\equiv (R_{k}:k=1,...,K)$, let $R_{-k}\equiv
(R_{k'}:k^{\prime }=1,2,...,k-1,k+1,...,K)$. For a pair of sub-vectors $R_{1},R_{2}$ in $R=(R_{1},R_{2})$, let $f_{R_{1}}$, $F_{R_{1}}$
denote the marginal density and distribution of a subvector, and let $%
f_{R_{2}|R_{1}}$, $F_{R_{2}|R_{1}}$\ denote the conditional density and
distribution. We write $f_{R_{2}|R_{1}=r_{1}}$ and $F_{R_{2}|R_{1}=r_{1}}$,
or simply $f(r_{2}|r_{1})$ and $F(r_{2}|r_{1})$, if there is need to be
specific about the realized value conditioned on.

\section{Test of Multiple Monotone Equilibria in Static Games}

\label{sec:static}

\subsection{The model}

Consider a simultaneous binary game with private information between a set
of players $\mathcal{I}$. Each player $i$ chooses $D_{i}\in \{1,0\}$. Let $X$
denote states that are common knowledge among players and reported in the
data. For each $i$, let $\epsilon _{i}\in \mathbb{R}$\ denote $i$'s private
information, a.k.a. \textquotedblleft type\textquotedblright ,
\textquotedblleft shock\textquotedblright\ or \textquotedblleft
signal\textquotedblright, which is unknown to other players and not
reported in the data. The ex post payoff for $i$ from $D_{i}=0$ is normalized to zero
while that from $D_{i}=1$ is:%
\begin{equation}
v_{i}(D_{-i},X)+\epsilon _{i},  \label{ex_post_payoff}
\end{equation}%
where $v_{i}$ is real-valued single index.  Additive separability is used later to establish an order among the equilibria for the game (see Lemma 1), but might be relaxed insofar as this ordering is preserved.

\begin{Assumption}
\label{assum:model_regularity} (i) Conditional on $X$, $\epsilon \equiv
(\epsilon _{i}:i\in \mathcal{I})$ is continuously distributed with bounded,
atomless density with respect to Lebesgue measure. (ii) The index $v_{i}$ is
bounded for all $i$.
\end{Assumption}

Assumption \ref{assum:model_regularity}\ guarantees the existence of a Bayes-Nash Equilibrium by Theorem 1 in \cite{athey2001} and allows for correlation between
private signals conditional on $X$. The test in \cite{dePaulaTang2012} cannot be applied in this case, because players' choices are correlated even
under the null of a single equilibrium in the data.

Conditional on any $x$, a \textit{pure strategy} for a player $i$ is a
function $\varepsilon \in \mathbb{R} \mapsto s_i(x,\varepsilon)\in \{0,1\}$. Let $U_{i}(x,\varepsilon _{i};s_{-i})$ denote the difference in player $i$'s
expected payoffs from choosing $1$ and $0$ when other players adopt pure
strategies $s_{-i}(x,\cdot)\equiv (s_{j}(x,\cdot):j\not=i)$. That is,%
\begin{equation*}
U_{i}(x,\varepsilon _{i};s_{-i})\equiv \int v_{i}(s_{-i}(x,\varepsilon
_{-i}),x)dF(\varepsilon _{-i}|\varepsilon _{i},x)+\varepsilon _{i},
\end{equation*}%
where $s_{-i}(x,\varepsilon _{-i})$ is shorthand for $(s_{j}(x,\varepsilon
_{j}):j\not=i)$. A \textit{pure-strategy Bayesian Nash equilibrium}
(p.s.BNE) in a game with states $x$ is a profile $(s_{i}^{\ast }(x,.):i\in 
\mathcal{I})$ such that%
\begin{equation}
s_{i}^{\ast }(x,\varepsilon _{i})=\mathbf{1}\left\{ U_{i}(x,\varepsilon
_{i};s_{-i}^{\ast })\geq 0\right\} \text{ for all }i\text{, }\varepsilon _{i}%
\text{.}  \label{BNE}
\end{equation}%
A \textit{p.s.BNE} conditional on states $x$ is \textit{monotone} if $%
s_{i}^{\ast }(x,.)$ is non-decreasing over the support of $\epsilon _{i}$
for each $i$. A monotone p.s.BNE conditional on $x$ is summarized by a
vector of thresholds $t(x)\equiv (t_{i}(x):i\in \mathcal{I})\in \mathbb{R}%
^{\#\mathcal{I}}$ such that $s_{i}^{\ast }(x,\varepsilon
_{i})=1\{\varepsilon _{i}\geq t_{i}(x)\}$. For brevity, we use the term
\textquotedblleft equilibrium\textquotedblright\ to refer to monotone
p.s.BNE when there is no ambiguity from the context below.

\begin{Assumption}
\label{assum:single_crossing} For any $x$ and $\varepsilon _{i}^{\prime
}>\varepsilon _{i}$, $U_{i}(x,\varepsilon _{i};s_{-i})\geq (>) 0$ implies $%
U_{i}(x,\varepsilon _{i}^{\prime };s_{-i})\geq (>) 0\ $whenever $s_{-i}$
consists of monotone strategies.
\end{Assumption}

Assumption \ref{assum:single_crossing}\ is the single-crossing condition in
games of incomplete information\ from \cite{athey2001} and holds automatically if private types are independent. Under Assumption \ref%
{assum:model_regularity}\ and \ref{assum:single_crossing}, Theorem 1 in
\cite{athey2001} implies that our game admits monotone p.s.BNE at any $x$. Let $%
\mathcal{T}(x)$ denote the complete set of threshold vectors that characterize an
equilibrium given $x$. That is, $\mathcal{T}(x)$ is the set of all vectors $%
t(x)$ such that $s_{i}^{\ast }(x,\varepsilon _{i})=1\{\varepsilon _{i}\geq
t_{i}(x)\}$, $i\in \mathcal{I}$ satisfies (\ref{BNE}).

\subsection{Overview of our method in a two-player game}

\label{subsec:overview}

A typical sample reports the states and choices from many games. The
identities of players may vary across the games; it is only maintained that
players' preference and information are drawn from the same data-generating
process (DGP) across these games. Player indexes such as $i$ and $j$ have
common meaning across the games in that the preference and information for
players with the same index from different markets are drawn from the same
distribution. For example, $i$ and $j$ could refer to Burger King and McDonald's across several markets, or husbands and wives across several households.  Our goal is to infer whether outcomes across these games are
generated from more than one equilibria.

To fix ideas, consider a model with two players $\mathcal{I}\equiv \{i,j\}$
across games that share the same realization of states $x$, which we
suppress in the notation $v_{i}(D_{-i})$ and throughout this subsection (except for the more general definition of adjancency) for simplicity. Assume players'
actions impose the same type of externalities on each other, that is, $%
sign[v_{i}(1)-v_{i}(0)]=$ $sign[v_{j}(1)-v_{j}(0)]$. (This allows for the possibility of multiple equilibria in this simple game.) $\mathcal{T}\subset 
\mathbb{R}^{\#\mathcal{I}}$ denotes the complete set of equilibria, and let $T\in 
\mathcal{T}$ denote the equilibrium selected in a game in the data.

For the rest of Section \ref{subsec:overview}, we elaborate on the idea
using a two-player model with $\mathcal{I}\equiv \{i,j\}$. First, we show this game admits a total order over the set of equilibria $\mathcal{T}$.

\begin{lemma}
\label{lm:total_order} Suppose Assumptions \ref{assum:model_regularity} and \ref%
{assum:single_crossing} hold, and $\#\mathcal{T}<\infty $.
Then $sign[(t_{i}-t_{i}^{\prime })(t_{j}-t_{j}^{\prime })]=$ $%
sign(v_{i}(1)-v_{i}(0))$ for all $t\not=t^{\prime }$ in $\mathcal{T}$.
\end{lemma}

\begin{proof}
Let $U_{i}(\varepsilon _{i};t_{j})$ denote $i$'s interim payoff when $j$
adopts a monotone pure strategy with a threshold $t_{j}$ with some abuse of
notation. That is, $U_{i}(\varepsilon _{i};t_{j})\equiv \int
v_{i}(1\{\varepsilon _{j}\geq t_{j}\})dF(\varepsilon _{j}|\varepsilon
_{i})+\varepsilon _{i}$. Under Assumptions \ref{assum:model_regularity} and %
\ref{assum:single_crossing}, $U_{i}(t_{i};t_{j})=0$ for all $t\equiv
(t_{i},t_{j})\in \mathcal{T}$. (Marginal players with private signals at
the equilibrium threshold receive zero interim payoff.) Suppose $%
v_{i}(0)>v_{i}(1)$. If $t_{j}>t_{j}^{\prime }$, then $U_{i}(t_{i};t_{j}^{%
\prime })<0$. This implies $U_{i}(\tilde{t}_{i};t_{j}^{\prime })<0$ for all $%
\tilde{t}_{i}<t_{i}$. Otherwise the single-crossing condition in Assumption %
\ref{assum:single_crossing} would fail. This then implies that $%
t_{i}^{\prime }>t_{i}$, because $U_{i}(t_{i}^{\prime };t_{j}^{\prime })$
needs to be zero for $t^{\prime }\in \mathcal{T}$. Hence $%
(t_{i}-t_{i}^{\prime })(t_{j}-t_{j}^{\prime })<0$. By a symmetric argument, $%
t_{i}^{\prime }<t_{i}$ if $t_{j}<t_{j}^{\prime }$. %
The case with $t_{j}\not=t_{j}'$ and $t_{i}=t_{i}^{\prime }$ are ruled out
under the stated conditions. To see this, suppose $t_{i}=t_{i}^{\prime }$.
Then the single-crossing condition in Assumption \ref{assum:single_crossing}%
\ and the finiteness of $\mathcal{T}$ implies $t_{j}=t_{j}^{\prime }$, which
means $t=t^{\prime }$. A symmetric argument proves the claim for the case
with $v_{i}(0)<v_{i}(1)$. $\ \ \ \ \square $\medskip
\end{proof}

Our goal is to test the null hypothesis that \textquotedblleft $T$\ is
degenerate in the data\textquotedblright\ (i.e., all observations in the data are drawn from a unique equilibrium) against the alternative that
\textquotedblleft $T$\ is stochastic in the data\textquotedblright . More specifically, let $F_{T}$ is the probability distribution over the equilibria in the game. We call $T$ degenerate if $F_{T}$ concentrates probability one on a particular equilibrium, and $T$ stochastic if $F_{T}$ allows for two or more equilibria to be selected with positive probability. To do
so, we exploit pairs of games whose equilibrium selection is related
under the alternative of multiple equilibria. We refer to such pairs as 
\textit{adjacent games} (we briefly reinstate the conditioning variables in the definition below for generality):

\begin{definition}
Let $m,n$ index two games.  When there are multiple equilibria, games $m$ and $n$ are \emph{adjacent} if $T_m$ and $T_n$ are not independent (conditional on $X_n, X_m$): $T_m \not\!\perp\!\!\!\perp T_n | X_n, X_m$.
\end{definition}

Further qualifications to such dependence are made for the results delineated in the subsections that follow.  Adjacency between two games may arise when equilibria are ``affiliated'' or equilibria are the same across two games. We provide such examples later in this section in more general settings.

Continuing on our two-player example, and once again suppressing the states $X$ which are presumed to be the same across games in this subsection, we make the following assumption:

\begin{Assumption}
\label{assum:CI} $\epsilon _{i,m}$ is independent from $(\epsilon
_{j,n},T_{n})$ conditional on $T_{m}$ for $m \ne n$.
\end{Assumption}

This assumption allows equilibrium selection to be correlated between
games, and thus allows these games to be adjacent. Under Assumption \ref{assum:CI}, we show that players'
choices from adjacent games ($D_{i,m}$ and $D_{j,n}$)\ are uncorrelated
under the null hypothesis that equilibrium selection is degenerate, i.e., outcomes are drawn from a single equilibrium, and generically correlated under the alternative that equilibrium selection is not degenerate.

To see this, decompose the covariance $C(D_{i,m},D_{j,n})$ between choices by two players $i$
and $j$ in adjacent games $m$ and $n$ as%
\begin{equation}
E[C(D_{i,m},D_{j,n}|T_{m},T_{n})]+C[E(D_{i,m}|T_{m},T_{n}),E(D_{j,n}|T_{n},T_{m})]%
\text{,}  \label{TC}
\end{equation}%
where $C(\cdot |\cdot )$ and $E(\cdot |\cdot )$ denote conditional
covariance and expectation respectively. Under Assumption \ref{assum:CI},%
\begin{eqnarray*}
E(D_{i,m}|D_{j,n},T_{m},T_{n}) &=&\Pr \{\epsilon _{i,m}\geq
T_{m,i}|1\{\epsilon _{j,n}\geq T_{n,j}\},T_{m},T_{n}\} \\
&=&\Pr \{\epsilon _{i,m}\geq T_{m,i}|T_{m}\} \\
&=& E(D_{i,m}|T_{m}),
\end{eqnarray*}%
where $T_{m,i}$ denotes the $i$-th component in the $(\#\mathcal{I})$-vector $T_{m}$. 
This implies $D_{i,m}$ is mean-independent from $D_{j,n}$ conditional on $T_n,T_m$.  
Thus $C(D_{i,m},D_{j,n}|T_{m},T_{n})=0$ and the first
term in (\ref{TC}) is zero regardless of whether the null is true. This implication
exploits the fact that, while private types are correlated within each game,
they are uncorrelated across adjacent games under Assumption \ref{assum:CI}.

Next, note that the second term in (\ref{TC}) is the covariance between $%
E(D_{i,m}|T_{m})$ and $E(D_{j,n}|T_{n})$ under Assumption \ref{assum:CI}.
Under the null, $T_{m}$ has a degenerate distribution across all games, so the covariance between $E(D_{i,m}|T_{m})$
and $E(D_{j,n}|T_{n})$, i.e., the second term in (%
\ref{TC}), is zero. 
Consequently, $C(D_{i,m},D_{j,n})=0$ under the null.
On the other hand, under the alternative, $%
E(D_{i,m}|T_{m})$ and $E(D_{j,n}|T_{n})$ are dependent and generally correlated as long as
the equilibrium selection is correlated between adjacent games. In such a
case, as we will show in Proposition 1, the sign of the covariance equals that of the interaction effect.

Given Lemma \ref{lm:total_order}, we define the following total order over $%
\mathcal{T}$: for all $t,t^{\prime }\in \mathcal{T}$, $t\geq t^{\prime }$\
if and only if $t_{i}\geq t_{i}^{\prime }$.\footnote{Note from the proof for Lemma 1 that $t_i=t_i'$ implies that $t_j=t_j'$.} We use this order to define
\textquotedblleft increasing\textquotedblright\ or \textquotedblleft
decreasing\textquotedblright\ functions over $\mathcal{T}$. We show that the
covariance between $D_{i,m}$ and $D_{j,n}$ is non-zero under the alternative
of multiple equilibria, and its sign is equal to the sign of the interaction
effect.

\begin{Assumption}
\label{assum:MES} (i) For any pair of adjacent games $n$ and $n$ with $m\not=n$ and any 
\textit{increasing function} $h$\textit{\ defined over }$\mathcal{T}$, $%
E[h(T_{n})|T_{m}=t]$\textit{\ is increasing in }$t$\textit{\ over }$\mathcal{%
T}$. (ii) $(\epsilon _{i,m},\epsilon _{j,m})$ is independent from $T_{m}$ in
each game $m$.
\end{Assumption}

Part (i) in Assumption \ref{assum:MES}\ states that there is some positive
association between equilibrium selection in adjacent games. It holds, for
example, if $T_{m}$ and $T_{n}$ are affiliated as defined in \cite{milgromWeber1982}. It is also satisfied trivially when adjacent games share the same
equilibrium, i.e., $\Pr \{T_{m}=T_{n}\}=1$ 
Part (ii) states that private
information does not interfere with equilibrium selection (once conditioning
on the states $x$ suppressed in notation).

\begin{proposition}
Suppose Assumptions \ref{assum:model_regularity}, \ref{assum:single_crossing}, \ref{assum:CI}, \ref{assum:MES} hold and $\#\mathcal{T}<\infty $. Then $%
C(D_{i,m},D_{j,n})\not=0$ for any adjacent games $m$ and $n$ with $m\not=n$ under the alternative, and its sign is the same as the
sign of $v_{i}(1)-v_{i}(0)$.
\end{proposition}

\begin{proof}
Recall from (\ref{TC}) that the covariance between $D_{i,m}$ and $D_{j,n}$
equals the covariance of $E(D_{i,m}|T_{m})$ and $E(D_{j,n}|T_{n})$ under 
Assumption \ref{assum:CI}. Applying the law of iterated expectations by
conditioning on $T_{m}$, we can write the covariance between $%
E(D_{i,m}|T_{m})$ and $E(D_{j,n}|T_{n})$\ as%
\begin{equation}
\int
[E(D_{i,m}|T_{m}=t)-E(D_{i,m})]\{E[E(D_{j,n}|T_{n})|T_{m}=t]-E(D_{j,n})%
\}dF_{T}(t)\text{,}  \label{cov_LIE}
\end{equation}%
where $F_{T}$ denote the marginal distribution of $T_{m}$, which is
non-degenerate under the alternative. Suppose $v_{i}(1)>v_{i}(0)$. Then
Lemma \ref{lm:total_order} implies that 
$t_{j}>t_{j}^{\prime }$\ whenever $t>t^{\prime }$ (recall this
total order is defined as $t_{i}>t_{i}^{\prime }$). Thus $h(t)\equiv $ $%
E(D_{j,n}|T_{n}=t)=$ $E(1\{\epsilon _{j,n}\geq t_{j}\}|T_{n}=t)$ is
decreasing in $t$ over $\mathcal{T}$ under (ii) in Assumption \ref{assum:MES}%
. Part (i) in Assumption \ref{assum:MES} then\ implies that $%
E[h(T_{n})|T_{m}=t]\equiv E(D_{j,n}|T_{m}=t)$ is decreasing in $t$ over $%
\mathcal{T}$. Next, note that $E(D_{i,m}|T_{m}=t)\equiv $ $\Pr \{\epsilon
_{i,m}\geq t_{i}|T_{m}=t\}$ is also decreasing in $t$ under (ii) in
Assumption \ref{assum:MES}. \ Therefore the covariance between $%
E(D_{i,m}|T_{m})$ and $E(D_{j,n}|T_{n})$ as expressed in (\ref{cov_LIE}) is
the covariance between two decreasing functions of the same latent variable $%
T_{m}$. Hence the covariance is positive. (See Theorem 2 in Schmidt (2003).)
A symmetric argument shows that\textit{\ }$C(D_{i,m},D_{j,n})<0$ when $%
v_{i}(0)>v_{i}(1)$. $\ \ \ \ \square $
\end{proof}

Given this proposition, it is possible to detect multiplicity in the data using the correlation in actions across adjacent games.  To do so, it is essential that researchers can match pairs of
adjacent games in the sample into clusters, within which the equilibrium
selection is known to be positively correlated. In practice, institutional details could be informative about how to construct such clusters. For example,
while studying the joint retirement decisions by married couples as
simultaneous Bayesian games, we may consider pairing households located in
geographic regions with similar demographics in a cluster. The underlying
rationale will be that regional institutional and cultural norms may induce association in the equilibrium selection in the joint retirement game for such couples.  It is worth mentioning that our method can be applied even when there is some form of weak dependence in equilibrium selection across clusters. In this case, our method applies as long as some version of the Law of Large Numbers allows us to pool across weakly dependent clusters and consistently estimate the covariance between $D_{im}$ and $D_{jn.}$

\subsection{Testable implications of multiple equilibria}

\label{subsec:general}

In this section we generalize the idea in Section \ref{subsec:overview} to a
full model with heterogeneous states and three or more players.  Recall that $T_{m}\equiv (T_{i,m}:i\in \mathcal{I})\in \mathcal{T}(X_{m})$
denotes the equilibrium selected in game $m$ in the data. Let $%
F_{T_{m},T_{n}|x,x^{\prime }}$ denote the\textit{\ joint} equilibrium
selection conditional on $(X_{m},X_{n})=(x,x^{\prime })$ in the DGP. That
is, for any $x$ and $x^{\prime }$, $F_{T_{m},T_{n}|x,x^{\prime }}$ is a
joint distribution with support $\mathcal{T}(x)\times \mathcal{T}(x^{\prime
})$. Our goal is to test the null that \textquotedblleft $%
F_{T_{m},T_{n}|x,x^{\prime }}$\ is degenerate for all $%
(x,x^{\prime })$\textquotedblright\ against the alternative that
\textquotedblleft $F_{T_{m},T_{n}|x,x^{\prime }}$\ is nondegenerate at least
for some $(x,x^{\prime })$\textquotedblright .

\begin{Assumption}
\label{assum:clusters} For any pair of adjacent games $m$ and $n$ with $m\not=n$, $\epsilon
_{i,m}$ is independent of $(\epsilon _{-i,n},T_{n},X_{n})$ conditional on $%
(T_{m},X_{m})$ for all $i\in \mathcal{I}$.
\end{Assumption}

Assumption \ref{assum:clusters}\ extends the conditional independence
condition (Assumption \ref{assum:CI}) in Section \ref{subsec:overview}\ by accounting for observed
states. It accommodates dependence between $(T_{m},T_{n},X_{m},X_{n})$, and
also allows for correlation between private types $\epsilon _{i,m}$, $\epsilon 
_{-i,m}$ and observed states $X_{m}$.

\begin{Assumption}
\label{assum:CES} \textit{There exists }$i,j\in \mathcal{I}$\textit{\ such
that the covariance between }$E(D_{i,m}|T_{m},X_{m})$\textit{\ and }$%
E(D_{j,n}|T_{n},X_{n})$\textit{\ is non-zero conditional on some }$%
(X_{m},X_{n})=(x,x^{\prime })$.
\end{Assumption}

 Assumption \ref{assum:CES}\
states there is nonzero correlation between equilibrium selection and thus implies adjacency of the two games. This extends an implication of Assumption \ref{assum:MES}\
in Section \ref{subsec:overview} to the case with heterogeneous states. It
does not hold if $(T_{m},X_{m})$ is independent from $(T_{n},X_{n})$ or if $X_n$ and $X_m$ are correlated but $T_m$ and $T_n$ are independent conditional on $X_m$ and $X_n$. To
further illustrate its content, we provide two examples of sufficient
conditions for Assumption \ref{assum:CES} later in this subsection.

\begin{proposition}
\label{pn:testable_imp} Suppose Assumptions \ref{assum:model_regularity}, %
\ref{assum:single_crossing}, \ref{assum:clusters} hold. Under the null, $%
C(D_{i,m},D_{j,n}|X_{m},X_{n})=0$ for all $i$, $j$. Under the alternative, $%
C(D_{i,m},D_{j,n}|X_{m},X_{n})\not=0$ for some $i,j$ and $%
(X_{m},X_{n})=(x,x^{\prime })$ if Assumption \ref{assum:CES}\ holds.
\end{proposition}

\begin{proof}[Proof of Proposition \protect\ref{pn:testable_imp}.]
Applying the Law of Total Covariance conditioning on $(T_{m},T_{n})$, we
have 
\begin{eqnarray}
C(D_{i,m},D_{j,n}|X_{m},X_{n}) &=&\underset{\equiv A}{\underbrace{%
C[E(D_{i,m}|T_{m},T_{n},X_{m},X_{n}),E(D_{j,n}|T_{m},T_{n},X_{m},X_{n})|X_{m},X_{n}]%
}}  \notag \\
&+&\underset{\equiv B}{\underbrace{%
E[C(D_{i,m},D_{j,n}|T_{m},T_{n},X_{m},X_{n})|X_{m},X_{n}]}}\text{.}
\label{eq:tot_cov}
\end{eqnarray}%
Monotone pure-strategy Bayesian Nash equilibria exist under Assumptions \ref%
{assum:model_regularity}\ and \ref{assum:single_crossing}. Assumption \ref%
{assum:clusters}\ implies that for $i,j$,%
\begin{eqnarray}
&&E(D_{i,m}|D_{j,n},T_{m},T_{n},X_{m},X_{n})  \label{eq:prod} \\
&=&\Pr \{\epsilon _{i,m}\geq T_{m,i}|1\{\epsilon _{j,n}\geq
T_{n,j}\},T_{m},T_{n},X_{m},X_{n}\}  \notag \\
&=&\Pr \{\epsilon _{i,m}\geq T_{m,i}|T_{m},X_{m}\} \notag = E(D_{i,m}|T_{m},X_{m}).  \notag
\end{eqnarray}%
Thus $D_{i,m}$ is mean-independent from $(D_{j,n},T_{n},X_{n})$ conditional on $%
(T_{m},X_{m})$. Therefore, the second term ``B" on the right-hand side of (%
\ref{eq:tot_cov}) is zero regardless of whether the null is true. Under
Assumption \ref{assum:clusters}, the first term ``A" equals the covariance
between $E(D_{i,m}|T_{m},X_{m})$ and $E(D_{j,n}|T_{n},X_{n})$ conditional on 
$(X_{m},X_{n})$. Under the null, $(T_{m},T_{n})$ are degenerate once
conditioning on $(X_{m},X_{n})$. Thus the term ``A" is zero under the null.
Under the alternative, the covariance between $E(D_{i,m}|T_{m},X_{m}=x)$ and 
$E(D_{J,n}|T_{n},X_{n}=x^{\prime })$ is non-zero for some $(x,x^{\prime })$
by Assumption \ref{assum:CES}.$\ \ \ \ \ \square $\medskip
\end{proof}

We conclude this section with examples of primitive conditions that imply
Assumption \ref{assum:CES}. The first generalizes affiliated equilibrium
selection in Section \ref{subsec:overview}\ to the case with heterogeneous observable covariates.\medskip

\noindent \textbf{Example 1.} (Affiliated Equilibrium Selection in Adjacent
Games) Consider a game with $\mathcal{I}\equiv \{i,j\}$ where both players
have the same sign of externality over each other under all states. That is, 
$sign(\Delta v_{i}(x))=sign(\Delta v_{j}(x))$ for all $x$, where $\Delta
v_{i}(x)\equiv v_{i}(1,x)-v_{i}(0,x)$. By conditioning the proof in Lemma %
\ref{lm:total_order}\ on any $x$ with $\#\mathcal{T}(x)<\infty $, we can
show that under Assumptions \ref{assum:model_regularity}, \ref{assum:single_crossing} and \ref{assum:clusters},%
\begin{equation}
sign[(t_{i}-t_{i}^{\prime })(t_{j}-t_{j}^{\prime })]=sign[\Delta v_{i}(x)]%
\text{ for all }t\not=t^{\prime }\text{ in }\mathcal{T}(x)\text{.}
\label{comove_thresh}
\end{equation}%
As noted in Section \ref{subsec:overview}, this allows us to define total
orders over $\#\mathcal{T}(x)$ for each $x$. Suppose for some $(x,x^{\prime
})$ with $sign(\Delta v_{i}(x))=sign(\Delta v_{i}(x^{\prime }))$ and $\#%
\mathcal{T}(x),\#\mathcal{T}(x^{\prime })<\infty $:
\medskip

\noindent \textbf{[Monotone Equilibrium Selection (MES)]}\textit{\ For some $x,x'$,}%
 $E[h(T_{n},X_{n})|T_{m}=t,X_{m}=x,X_{n}=x^{\prime }]$\textit{\ is increasing
in }$t$\textit{\ over }$\mathcal{T}(x)$\textit{\ for any }$h$\textit{\ that
is increasing in its first argument. }

This condition generalises Assumption \ref{assum:MES} by allowing for heterogeneous states. It holds, for example, if $T_{m}$ and $T_{n}$ are
affiliated conditional on $(X_{m},X_{n})=(x,x^{\prime })$ as defined in
Milgrom and Weber (1982). As in Section \ref{subsec:overview}, we maintain
that private types does not affect equilibrium selection directly,
i.e., $\epsilon _{i}$ is independent from $T$ given $X$.

We show that the MES condition implies Assumption \ref{assum:CES}. 
Therefore, it also implies $%
C(D_{i,m},D_{j,n}|X_{m}=x,X_{n}=x^{\prime })\not=0$ under the alternative.
Let $E(D_{i}|T=t,X=x)\equiv \phi _{i}(t,x)$ and $E(D_{j}|T=t,X=x^{\prime
})\equiv \phi _{j}(t,x^{\prime })$. The Law of Iterated Expectation implies
that the covariance between $E(D_{i,m}|T_{m},X_{m})$\ and $%
E(D_{j,n}|T_{n},X_{n})$\ conditional on $(X_{m},X_{n})=(x,x^{\prime })$ is:%
\footnote{%
To derive this expression, consider a generic random vector $(W,Y,Z)$. Let $%
\mu _{f}(z)\equiv E[f(W,Z)|Z=z]$ and $\mu _{g}(z)\equiv E[g(Y,Z)|Z=z]$. Then 
$C(f(W,Z),g(Y,Z)|Z=z)$ equals 
\begin{eqnarray*}
&&E\{[f(W,z)-\mu _{f}(z)][g(Y,z)-\mu _{g}(z)]|Z=z\}=E[E\{\left. [f(W,z)-\mu
_{f}(z)][g(Y,z)-\mu _{g}(z)]\right\vert Y,Z=z\}|Z=z] \\
&=&E\left( \{E[f(W,z)|Y,Z=z]-\mu _{f}(z)\}\times \lbrack g(Y,z)-\mu
_{g}(z)]\mid Z=z\right) \text{.}
\end{eqnarray*}%
The claim in the text holds with $W\equiv T_{n},Y\equiv T_{m},$ $%
f(W,Z)\equiv E(D_{j,n}|T_{n},X_{n})$, $g(Y,Z)\equiv E(D_{i,m}|T_{m},X_{m})$
and $Z\equiv (X_{m},X_{n})$.}%
\begin{equation}
\int \left( 
\begin{array}{c}
\{\phi _{i}(t,x)-E[\phi _{i}(T_{m},x)|x,x^{\prime }]\}\times \\ 
\{E[\phi _{j}(T_{n},x^{\prime })|T_{m}=t,x,x^{\prime }]-E[\phi
_{j}(T_{n},x^{\prime })|x,x^{\prime }]\}%
\end{array}%
\right) dF_{T_{m}|X_{m},X_{n}}(t|x,x^{\prime })\text{.}  \label{cond_covar}
\end{equation}%
Here the expectation $E[\phi _{i}(T_{m},x)|x,x^{\prime }]$ is with respect
to $T_{m}$ given $(X_{m},X_{n})=(x,x^{\prime })$. First, consider the case
where $\Delta v_{i}(x)>0$ and $\Delta v_{i}(x^{\prime })>0$. By (\ref%
{comove_thresh}), we can define total orders over $\mathcal{T}(x)$ and $%
\mathcal{T}(x^{\prime })$ as \textquotedblleft $t>t^{\prime }$\ if and only
if $t_{i}>t_{i}^{\prime }$\textquotedblright . Because $\epsilon _{i}$ is
independent from $T$ given $X$, $\phi _{i}(t,x)=\Pr \{\epsilon _{i}\geq
t_{i}|X=x\}$ is decreasing in $t$ over $\mathcal{T}(x)$. Besides, $\phi
_{j}(t,x^{\prime })$ is also decreasing in $t$ over $\mathcal{T}(x^{\prime
}) $, because (\ref{comove_thresh}) implies that $t_{j}>t_{j}^{\prime }$\
whenever $t_{i}>t_{i}^{\prime }$. The MES condition then implies $E[\phi
_{j}(T_{n},x^{\prime })|T_{m}=t,x,x^{\prime }]$ is decreasing in $t$. Hence (%
\ref{cond_covar}) is the covariance between two decreasing functions of $%
T_{m}$ conditional on $(X_{m},X_{n})=(x,x^{\prime })$, and is positive. (See
Theorem 2 in \cite{schmidt2003}.) A symmetric argument shows that (\ref%
{cond_covar}) is negative when $\Delta v_{i}(x)<0$ and $\Delta
v_{i}(x^{\prime })<0$. Hence Assumption \ref{assum:CES}\ holds for adjacent games with $X_{m}=x$
and $X_{n}=x^{\prime }$. $\ \ \square $\medskip

In the second example, the adjacent games share the same equilibrium
selection with probability one, and Assumption \ref{assum:CES}\ is
satisfied.\medskip

\noindent \textbf{Example 2. }(Same Equilibrium in Adjacent Games) Consider
a game with three players or more ($\#\mathcal{I}\geq 3$). Assumption \ref%
{assum:CES}\ holds if adjacent games are known to select the same
equilibrium:\medskip

\noindent [\textbf{Same Equilibrium Selection (SES)}] $\Pr \{T_{m}=T_{n}|X_{m}=X_{n}\}=1$\textit{\ for
any two adjacent games }$m$\textit{\ and }$n$\textit{.}\medskip

\noindent The SES condition accommodates multiple equilibria across disjoint
pairs or clusters of adjacent games in the data.\footnote{It is worth mentioning that one can also use this condition to test multiple Markovian Perfect equilibria (MPE) in  dynamic games with zero discount factors, where private types are correlated through unobserved heterogeneity that is serially independent. In this case, one can treat the ``stage games'' from each period in the same Markov process as static games within the same ``cluster'', and the SES condition simply means players do not switch between MPEs within a process.} 
Under this condition, we
can define $\tilde{F}_{T|x}(.)$ as the marginal distribution of equilibria
across pairs of adjacent games with $(X_{m},X_{n})=(x,x)$. The conditional
covariance between $D_{i,m}$ and $D_{i,n}$ is:\footnote{%
Recall that agents labelled by the same index\ across games are not required
to be the same exact individual. It is only maintained that these
individuals with identical indexes share the same preference, and are drawn
from the same component in the DGP.}%
\begin{equation*}
\int \left( 
\begin{array}{l}
\{\phi _{i}(t,x)-E[\phi _{i}(T_{m},x)|X_{m}=x,X_{n}=x]\}\times \\ 
\{\phi _{i}(t,x)-E[\phi _{i}(T_{n},x)|X_{m}=x,X_{n}=x]\}%
\end{array}%
\right) d\tilde{F}_{T|x}(t)\text{,}
\end{equation*}%
which is the variance of $\phi _{i}(T,x)$ with $T$ drawn from $\tilde{F}%
_{T|x}$. Such a variance is strictly positive under the alternative because $%
\tilde{F}_{T|x}$ is non-degenerate. Hence Assumption \ref{assum:CES} holds
with $i=j$ and $(X_{m},X_{n})=(x,x)$. We can also replace $D_{i,n}$ with $%
D_{j,n}$ in this example. In this case, the covariance is also nonzero
generically, because $\phi _{i}(T,x)$ and $\phi _{j}(T,x)$ are functions of
the same variable. \ $\ \ \square $

Building on these ideas, one can implement a statistical test for the null hypothesis of single equilibrium in the DGP, using tools that already exist in the econometrics literature. Let $\{m,n\}$ label two adjacent games. With $%
D_{i,m},D_{j,n}$ being binary, zero conditional covariance is equivalent to conditional mean-independence.\footnote{Abstracting from covariates for simplicity and letting $p_{i,m}=E(D_{i,m})$ and $p_{j,n}=E(D_{j,n})$, this obtains since $C(D_{i,m},D_{j,n})=E[(D_{i,m}-p_{i,m})(D_{j,n}-p_{j,n})]=E[(E(D_{i,m}|D_{j,n})-p_{i,m})(D_{j,n}-p_{j,n})]$ by the Law of Iterated Expectations.  Further developing this expression we get that it is equal to $p_{j,n}[E(D_{i,m}|D_{j,n}=1)-E(D_{i,m})](1-p_{j,n})+(1-p_{j,n})[E(D_{i,m}|D_{j,n}=0)-E(D_{i,m})](0-p_{j,n})=p_{j,n}(1-p_{j,n})[E(D_{i,m}|D_{j,n}=1)-E(D_{i,m}|D_{j,n}=0)]$.  If the covariance is zero, then $E(D_{i,m}|D_{j,n}=1)-E(D_{i,m}|D_{j,n}=0)=E(D_{i,m})$.}  The null hypothesis of a unique equilibrium in the data, formulated as a zero conditional covariance restriction above, is thus equivalent to:
\begin{equation}
H_{0}:E\left( D_{i,m}|X_{m},X_{n},D_{j,n}\right) =E\left(
D_{i,m}|X_{m},X_{n}\right) \text{ a.e. }  \label{newNull}
\end{equation}%
This equality fails with positive probability under the alternative of
multiple equilibria.  Significance tests in
such nonparametric models have been studied extensively in the literature,
both for continuous and discrete covariates. See, for example, \cite{fanLi1996}, \cite{racine1997}, \cite{chenFan1999}, \cite{delgadoManteiga2001}, \cite{lavergne2001}, and more recently, in \cite{racine_et_al2006}.

\subsection{Relation to literature}

\cite{aguirregabiriaMira2019} studied static games with incomplete
information when players' private information are independent conditional on some unobserved heterogeneity known to all players but not measured in the sample. To identify the full model in its generality, they treated
equilibrium selection as a component in the game-level unobserved
heterogeneity. They proposed an eigendecomposition method to recover conditional choice probabilities (CCPs) given unobserved heterogeneity, and then used exclusion restrictions similar to \cite{bajari_et_al2010} to identify players' ex post payoff functions from these CCPs.

In their Section 4.3, \cite{aguirregabiriaMira2019} proposed a novel
idea to test multiple equilibria by checking whether the payoff functions
recovered from the component CCPs in the finite mixture leads to distinctive
payoff functions. The main idea is that, if there are multiple equilibria,
then some of the component CCPs recovered from the finite mixture would lead
to the same payoff functions. To implement this idea, one needs to follow
sequential steps: (1) recover component CCPs given unobserved heterogeneity
and equilibrium selection using eigendecomposition, (2) estimate
payoffs from component CCPs, and (3) comparing payoff estimates calculated
from different CCPs. The first two steps involve nonparametric estimation,
and the last step requires developing asymptotic theory of a test statistic
that properly accounts for the estimation error in the first two steps.

Unlike \cite{aguirregabiriaMira2019}, our goal in this paper is not to
fully identify discrete Bayesian games with unobserved heterogeneity and
multiple equilibria. Rather, our objective is to construct a robust test for
multiple equilibria in settings where players' information are generally
correlated. These include, but are not restricted to, the case considered in
\cite{aguirregabiriaMira2019} where players' private signals are
independent conditional on game-level unobserved heterogeneity. Moreover,
our method is easy to implement. It formulates the task of inferring
multiple equilibria as a test for the relevance of discrete covariates in
nonparametric expectation. Thus one can directly tap into existing methods
in the literature mentioned above to define a test statistic using kernel regressions, and characterize its asymptotic properties.

\subsection{Numerical examples}

In our simulation exercises, we investigate how the covariance between
choices varies across different designs of static Bayesian games with binary
choices. We first report results from games with two players $\{i,j\}$. In
all designs, the vector of observed states $X\in \mathbb{R}^{2}$ consists of
a discrete $X_{1}$, which is uniformly distributed over a discrete support $%
\{1,2,3,4\}$, and a continuous $X_{2}$, which is standard uniform over $[0,1]
$. We experiment with \textit{three} specifications of ex post payoffs for
player $i$ from choosing action $1$:
\newline
\newline
\noindent Specification 1. $\beta _{1}X_{1}+\beta _{2}X_{2}+\delta
D_{j}+\epsilon _{i}$.

\noindent Specification 2. $\beta _{1}X_{1}+\beta _{2}X_{2}+\beta
_{3}X_{1}X_{2}+\beta _{4}X_{2}^{2}+\delta D_{j}+\epsilon _{i}$.

\noindent Specification 3. $\beta _{1}X_{1}+\beta _{2}X_{2}+\beta
_{3}X_{1}X_{2}+\beta _{4}X_{2}^{2}+\beta _{5}\sqrt{X_{1}}+\delta
D_{j}+\epsilon _{i}$.
\newline
\newline
\noindent Ex post payoff for player $j$ in each design is specified in a
similar way, with subscripts $i$ and $j$ swapped in $D_{j}$ and $\epsilon
_{i}$. The slope coefficients are $\beta _{1}=1/4$, $\beta _{2}=1/5$, $\beta
_{3}=1/10$, $\beta _{4}=-1/5$, $\beta _{5}=-1/10$ and $\delta =-2$. The pair
of private information components follows a bivariate normal distribution with zero
mean, a unit variance, and a correlation coefficient $\rho \not=0$. For each
specification of ex post payoffs, we experiment with two designs with $\rho
=0.5$ and $0.7$ respectively. For each design, we solve for multiple
pure-strategy BNE, characterized as paired thresholds conditional on $X$%
.\medskip 

Given these parameter values and using the law of total covariances (conditioning on equilibrium selection), we can calculate the
covariance between $D_{i,m}$ and $D_{j,n}$ from adjacent games conditional
on state $X$. Table 1 reports the covariance between $D_{i,m}$ and $D_{j,n}$%
, averaged over the support of $X$, for various designs of the
data-generating process. The signs of these covariances are consistent with
the negative interaction effect $\delta <0$. As the data-generating process
moves farther way from the null hypothesis of single equilibrium ($\varphi =0 $ with $\varphi$ denoting the probability for mixing between two equilibria), the magnitude of these covariances also increases.

\begin{center}
Table 1. Covariance between $D_{i,m}$ and $D_{j,n}$ in 2-player
games\medskip 

\begin{tabular}{c|ccc|ccc}
\hline
& \multicolumn{3}{|c|}{$\rho =0.50$} & \multicolumn{3}{|c}{$\rho =0.70$} \\ 
\hline
$\varphi $ & Spec\#1 & Spec\#2 & Spec\#3 & Spec\#1 & Spec\#2 & Spec\#3 \\ 
\hline
0.1 & -0.0303 & -0.0300 & -0.0286 & -0.0366 & -0.0364 & -0.0348 \\ 
0.2 & -0.0538 & -0.0534 & -0.0509 & -0.0651 & -0.0647 & -0.0618 \\ 
0.3 & -0.0706 & -0.0701 & -0.0668 & -0.0854 & -0.0849 & -0.0812 \\ 
0.4 & -0.0807 & -0.0801 & -0.0764 & -0.0976 & -0.0970 & -0.0928 \\ 
0.5 & -0.0840 & -0.0835 & -0.0796 & -0.1017 & -0.1010 & -0.0966 \\ \hline
\end{tabular}%
\medskip 
\end{center}

Next, we do a similar investigation for games with three players indexed by $%
i,j,k$. The ex post payoff for player $1$ from choosing action $i$ is
specified as 
\begin{equation*}
X\beta _{i}+\delta _{ij}D_{j}+\delta _{ik}D_{k}+\epsilon _{k}\text{;}
\end{equation*}%
and likewise for the other two players $j$ and $k$ with coefficients $(\beta
_{j},\delta _{ji},\delta _{jk})$ and $(\beta _{k},\delta _{ki},\delta _{kj})$
respectively. The vector of private signals $(\epsilon _{i},\epsilon
_{j},\epsilon _{k})$ is tri-variate normal with zero mean and unit variance.
The correlation coefficients are $\rho _{ij}=0.75$ and $\rho _{ik}=\rho
_{jk}=0.8$. For simplicity, we let $\beta _{i}=\beta _{j}$, $\delta
_{ki}=\delta _{kj}=-3.1$, $\delta _{ik}=\delta _{jk}=-3.25$ and $\delta
_{ij}=\delta _{ji}=-0.95$ in all designs. Thus players $i$ and $j$ are ex
ante identical in all designs. The observed states $X=(X_{1},X_{2})$ follows
the same distribution as in two-player cases above. We also experiment with
three specifications of the index $X\beta _{i}$ in ex post payoffs, with $%
\beta _{i}=\beta _{j}=(1.73,-0.97,0.155,-0.16,-0.01)$ and $\beta
_{k}=(1.91,-1.645,-0.295,0.29,0.75)$, where the five components correspond to $(\beta_1,...,\beta_5)$ in Specification 1-3 above.
Similar to the case with two players,
our simulation uses DGPs that mix between extreme equlibria (in terms of
threshold for $D_{k}=1$), with adjacent games sharing the same equilibrium
selection.\medskip 

\begin{center}
Table 2. Analytic covariances in 3-player games\medskip 

\begin{tabular}{c|ccc|ccc}
\hline
& \multicolumn{3}{|c}{$Cov(D_{k,m},D_{k,n})$} & \multicolumn{3}{|c}{$%
Cov(D_{k,m},D_{j,n})$} \\ \hline
$\varphi $ & Spec\#1 & Spec\#2 & Spec\#3 & Spec\#1 & Spec\#2 & Spec\#3 \\ 
\hline
\multicolumn{1}{c|}{0.1} & 0.0441 & 0.0529 & 0.0335 & -0.0261 & -0.0275 & 
\multicolumn{1}{c}{-0.0218} \\ 
\multicolumn{1}{c|}{0.2} & 0.0785 & 0.0941 & 0.0596 & -0.0464 & -0.0489 & 
\multicolumn{1}{c}{-0.0388} \\ 
\multicolumn{1}{c|}{0.3} & 0.1030 & 0.1235 & 0.0782 & -0.0609 & -0.0641 & 
\multicolumn{1}{c}{-0.0509} \\ 
\multicolumn{1}{c|}{0.4} & 0.1177 & 0.1412 & 0.0894 & -0.0696 & -0.0733 & 
\multicolumn{1}{c}{-0.0582} \\ 
\multicolumn{1}{c|}{0.5} & 0.1226 & 0.1470 & 0.0931 & -0.0725 & -0.0763 & 
\multicolumn{1}{c}{-0.0606} \\ \hline
\end{tabular}%
\medskip 
\end{center}

Similar to the two-player cases, the magnitude of these covariances,
averaged over the state space, increases as the data-generating process
moves farther from the null hypothesis (that is, as $\varphi $ increases between 0 and 0.5).
The sign of the covariance of the same-type players $D_{k,m},D_{k,n}$ is
positive, while that between different types of players is negative.\footnote{Though we do not have analytical results for this pattern with three or more players, this is what one would expect to obtain in two-player games with strategic substitution (i.e., negative interaction parameters).} 

\section{Test of Multiple Equilibria in Dynamic Games}

\label{sec:dynamic}

Dynamic games with private information has been studied extensively in the
literature. See \cite{aguirregabiriaMira2007}, \cite{bajari_et_al2007}, \cite{pesendorferDengler2008} and \cite{arcidiaconoMiller2011} for example. As in the case with static games, presence of
multiple Markov Perfect equilibria in the sample poses challenges to estimation and
inference based on conditional choice probabilities.

We introduce new tests for multiple equilibria in the dynamic games with
incomplete information, where players' information are correlated in any
given period through a time-varying, serially correlated unobserved
heterogeneity. The idea for the test is related to that in the static case
in that it also amounts to investigating the correlation between decisions
made in \textquotedblleft adjacent\textquotedblright\ games. Specifically,
in the dynamic setting, different periods play the role of adjacent games and we form adjacent pairs of actions by matching
players across different periods but within the same dynamic game.

We also note that the ideas below can be used for static games if one assumes that the discount factor is zero and takes a `cluster' to correspond to one such dynamic game (with zero discount factor).

\subsection{The model \label{sec:dynamicModel}}

Let $\mathcal{I}$ denote a set of players. Each $i\in \mathcal{I}$ makes a
sequence of discrete actions $D_{i,t}\in \mathcal{D}$ indexed by time
periods $t=1,2,...\infty $. In each period $t$, player $i$ observes a vector
of \textit{common} states $(X_{t},\xi _{t})$ and \textit{private} shocks $%
\epsilon _{i,t}\equiv (\epsilon _{i,t}^{d}:d\in \mathcal{D})$. While $X_{t}$
is reported in the sample, $\xi _{t}$ is a time-varying discrete unobserved game 
heterogeneity (DUH) not recorded in the data. Let $D_{t}\equiv (D_{i,t}:i\in 
\mathcal{I})$ and $\epsilon _{t}\equiv
(\epsilon _{i,t}:i\in \mathcal{I})$ for each $t$.

A player's payoff at time $t$ is given by a real-valued function $\pi
_{i}(D_{t},X_{t},\xi _{t},\epsilon _{i,t})$. In each period $t$, players
make choices simultaneously to maximize%
\begin{equation*}
E\left\{ \sum\nolimits_{s=t}^{\infty }\beta ^{s-t}\pi _{i}(D_{s},X_{s},\xi
_{s},\epsilon _{i,s})\mid x_{t},\xi _{t},\varepsilon _{i,t}\right\} ,
\end{equation*}%
where $\beta \in (0,1)$ is a constant discount factor known to all players.
We maintain that $(X_{t},\xi _{t},\epsilon _{t})$ follows a controlled
first-order Markov process with a time-homogeneous transition density $h$
that satisfies the following conditional independence: 
\begin{equation}
h(X_{t+1},\xi _{t+1},\epsilon _{t+1}|X_{t},\xi _{t},\epsilon
_{t},D_{t})=g(\epsilon _{t+1}|X_{t+1},\xi _{t+1})f(X_{t+1},\xi
_{t+1}|X_{t},\xi _{t},D_{t})\text{,}  \label{ci}
\end{equation}%
where the private signals are independent conditional on common states:%
\begin{equation*}
g(\epsilon _{t+1}|X_{t+1},\xi _{t+1})=\prod\nolimits_{i\in \mathcal{I}%
}g_{i}(\epsilon _{i,t+1}|X_{t+1},\xi _{t+1})\text{,}
\end{equation*}%
with $g_{i}(\cdot \mid \cdot ,\cdot )$ denoting a conditional marginal
density of private signals. Note that this specification of the law of
transition permits the contemporary private signals $\epsilon
_{i,t},\epsilon _{j,t}$ to be correlated through unobserved $\xi _{t}$ when
conditioning on $X_{t}$ alone. In what follows, we suppress the time
subscript $t$ in $x_{t},\xi _{t},D_{i,t},D_{t}$ and $\varepsilon _{i,t}$ to
simplify notation.

Let $\sigma _{i}(x,\xi ,\varepsilon _{i})\rightarrow \mathcal{D}$ denote a
pure Markovian strategy adopted by a player $i$; and let $\sigma \equiv
(\sigma _{i}:i\in \mathcal{I})$. Let $p_{i}^{(\sigma )}(d_{-i}|x,\xi )$
denote the probability for $D_{-i}\equiv (D_{j}:j\in \mathcal{I}\backslash
\{i\})=d_{-i}$ conditional on $i$'s information $(x,\xi )$ when the strategy
profile is $\sigma $;\footnote{%
Note $\varepsilon _{i}$ is not in $i$'s information set due to
independence between $\varepsilon _{i}$ and $\varepsilon _{-i}$ conditional
on $(x,\xi )$.} and let 
\begin{equation*}
f_{i}^{(\sigma )}(x^{\prime },\xi ^{\prime }|x,\xi ,d_{i})\equiv
\sum\nolimits_{d_{-i}}p_{i}^{(\sigma )}(d_{-i}|x,\xi )f(x^{\prime },\xi
^{\prime }|x,\xi ,d_{i},d_{-i})
\end{equation*}%
denote the transition matrix given $i$'s information. Let $\tilde{V}%
_{i}^{(\sigma )}(x,\xi ,\varepsilon _{i})$ denote payoff for player $i$ if
it behaves optimally now and onwards given other firms' strategies in $%
\sigma $. By the Bellman's principle of optimality, 
\begin{eqnarray*}
&&\tilde{V}_{i}^{(\sigma )}(x,\xi ,\varepsilon _{i}) \\
&=&\max_{d_{i}\in \mathcal{D}}\left\{ \Pi _{i}^{(\sigma )}(d_{i},x,\xi
,\varepsilon _{i})+\beta \int \int \left[ \int \tilde{V}_{i}^{(\sigma
)}(x^{\prime },\xi ^{\prime },\varepsilon _{i}^{\prime })g_{i}(\varepsilon
_{i}^{\prime }|x^{\prime },\xi ^{\prime })d\varepsilon _{i}^{\prime }\right]
f_{i}^{(\sigma )}(x^{\prime },\xi ^{\prime }|x,\xi ,d_{i})d\xi ^{\prime
}dx^{\prime }\right\}
\end{eqnarray*}%
where $g_{i}(\varepsilon _{i}|x,\xi )$ denotes the conditional distribution
of $i$'s private shocks and%
\begin{equation*}
\Pi _{i}^{(\sigma )}(d_{i},x,\xi ,\varepsilon _{i})\equiv
\sum\nolimits_{d_{-i}}p_{i}^{(\sigma )}(d_{-i}|x,\xi )\pi
_{i}(d_{-i},d_{i},x,\xi ,\varepsilon _{i})
\end{equation*}%
is $i$'s expected payoff in time $t$ given other players' strategies in $%
\sigma $. Define an integrated value function for $i$ as:%
\begin{equation*}
V_{i}^{(\sigma )}(x,\xi )\equiv \int \tilde{V}_{i}^{(\sigma )}(x,\xi
,\varepsilon _{i})g_{i}(\varepsilon _{i}|x,\xi )d\varepsilon _{i}\text{.}
\end{equation*}%
For a given strategy profile $\sigma $, the integrated value function is
characterized as the unique solution to the following fixed point equation:%
\footnote{%
It can be shown by verifying the Blackwell sufficient conditions that, for
any given $i$ and $\sigma $, the right-hand side is a contraction mapping in
the functional space for integrated value functions.}%
\begin{eqnarray*}
&&V_{i}^{(\sigma )}(x,\xi ) \\
&=&\int \max_{d_{i}\in \mathcal{D}}\left\{ \Pi _{i}^{(\sigma )}(d_{i},x,\xi
,\varepsilon _{i})+\beta \int \int V_{i}^{(\sigma )}(x^{\prime },\xi
^{\prime })f_{i}^{(\sigma )}(x^{\prime },\xi ^{\prime }|x,\xi ,d_{i})d\xi
^{\prime }dx^{\prime }\right\} g_{i}(\varepsilon _{i}|x,\xi )d\varepsilon
_{i},\text{\ }\forall x\in \mathcal{X}\text{.}
\end{eqnarray*}%
A \textit{Markov Perfect Equilibrium} (MPE) is a profile of strategies $%
\sigma ^{\ast }$ such that for any $i$ and $(x,\varepsilon _{i})$,%
\begin{equation*}
\sigma _{i}^{\ast }(x,\xi ,\varepsilon _{i})=\arg \max_{d_{i}\in \mathcal{D}%
}\left( \Pi _{i}^{(\sigma ^{\ast })}(d_{i},x,\xi ,\varepsilon _{i})+\beta
\int \int V_{i}^{(\sigma ^{\ast })}(x^{\prime },\xi ^{\prime
})f_{i}^{(\sigma ^{\ast })}(x^{\prime },\xi ^{\prime }|x,\xi ,d_{i})d\xi
^{\prime }dx^{\prime }\right)
\end{equation*}%
for all $i\in \mathcal{I}$ and $(x,\xi ,\varepsilon _{i})$. In general, the
model admits multiple MPE. Our goal is to test the null hypothesis that the
data-generating process is rationalized by a single MPE.

\subsection{Detecting multiple equilibria \label{sec:MPE}}

Consider a sample that consists of a large number of independent dynamic
games. For each game, the sample reports states and choices over a finite
number of periods $t=1,2,...,T$. We maintain the following assumption about
the data-generating process.

\begin{Assumption}
\label{assum:fix_mpe} Within each game, the players follow a fixed profile
of MPE strategies $\sigma ^{\ast }$ in all periods $t=1,2,...,T$.
\end{Assumption}

To fix ideas, consider a simplified model with no states $X_{t}$. An MPE in
this case is a profile of functions $(\sigma _{i}:i\in \mathcal{I})$, with
each $\sigma _{i}$ mapping from $(\xi _{t},\epsilon _{i,t})$ to $D_{i,t}$.
Let $\mathcal{M}$ denote the discrete and finite set of MPE admitted by the
model, and let $M\equiv \#\mathcal{M}$ denote its cardinality. The
equilibrium selection mechanism $\varphi (\cdot )$ is a probability mass
function with support $\mathcal{M}$. Let $K$ denote the cardinality of the
support of $\xi _{t}$, and let $\lambda (\cdot )$ denote the probability
mass of $\xi _{t}$.\footnote{%
In general $\lambda (\cdot )$ need not be time-homogeneous. We suppress such
dependence to simplify notation.}

Let's partition the set of players $\mathcal{I}$ into $\mathcal{I}_{A}%
\mathcal{\cup I}_{B}$, and define 
\begin{equation*}
D_{t}\equiv (D_{A,t},D_{B,t})\text{ with }D_{A,t}\equiv (D_{i,t})_{i\in 
\mathcal{I}_{A}}\text{ and }D_{B,t}\equiv (D_{j,t})_{j\in \mathcal{I}_{B}}.
\end{equation*}%
Let $\delta _{A}\equiv (\#\mathcal{D})^{\#\mathcal{I}_{A}}$ and $\delta
_{B}\equiv (\#\mathcal{D})^{\#\mathcal{I}_{B}}$ denote the cardinality of
the support of $D_{A,t}$ and $D_{B,t}$

In what follows, let $\sum_{m}$ and $\sum_{\xi }$ denote the summation of $m$
and $\xi $ over their respective support. Using the law of total
probability, decompose the joint probability mass of $D_{t}$ as%
\begin{eqnarray}
P(D_{t}) &=&\sum\nolimits_{m}\varphi (m)\left[ \sum\nolimits_{\xi }\lambda
(\xi )P_{m}(D_{A,t},D_{B,t}|\xi )\right]  \label{PD} \\
&=&\sum\nolimits_{m,\xi }\varphi (m)\lambda (\xi )P_{m}(D_{A,t}|\xi
)P_{m}(D_{B,t}|\xi ),  \notag
\end{eqnarray}%
where $P_{m}(\cdot |\xi )$ denotes the conditional probability mass implied
by a single MPE $m\in \mathcal{M}$ conditional on $\xi $.

For convenience, let $\{\omega _{j}\}_{j=1,..,J}$ denote elements of the
joint support of the discrete vector $(m,\xi _{t})$. By construction, $%
J\equiv MK$. Let $\theta_{j}\equiv \Pr \{(m,\xi _{t})=\omega _{j}\}$. In matrix
notation, the joint probability mass of $(D_{A,t},D_{B,t})$ is summarized by 
\begin{equation}
\underset{Q_{A}}{\mathbf{P}\equiv \underbrace{\left( 
\begin{array}{ccc}
q_{A,1} & \cdots & q_{A,J}%
\end{array}%
\right) }}\underset{\Lambda }{\underbrace{\left( 
\begin{array}{ccc}
\theta_{1} & 0 & 0 \\ 
0 & \ddots & 0 \\ 
0 & 0 & \theta_{J}%
\end{array}%
\right) }}\underset{Q_{B}^{\prime }}{\underbrace{\left( 
\begin{array}{c}
q_{B,1}^{\prime } \\ 
\vdots \\ 
q_{B,J}^{\prime }%
\end{array}%
\right) }=}\sum\nolimits_{j}\theta_{j}q_{A,j}q_{B,j}^{\prime }\text{,}
\label{decomp1}
\end{equation}%
where $q_{A,j}$ is a $\delta _{A}$-by-$1$ column vector $\{P_{m}(D_{A,t}=a%
\mid \xi _{t}):a\in \mathcal{D}_{A}\}$ with $(m,\xi _{t})=\omega _{j}$, and $%
q_{B,j}$ is a $\delta _{B}$-by-$1$ column vector $\{P_{m}(D_{B,t}=b\mid \xi
_{t}):b\in \mathcal{D}_{B}\}$ with $(m,\xi _{t})=\omega _{j}$.

\begin{Assumption}
\label{assum:FR} $\min \{\delta _{A},\delta _{B}\}>MK$; and both $Q_{A}$ and 
$Q_{B}$ have full rank.
\end{Assumption}

The rank condition requires that the vectors of conditional choice
probabilities in $\mathbf{P}$ are not linearly dependent. Such conditions
are common in structural models with discrete unobserved heterogeneity. For
example, Assumption 2.1 and 2.2 in Hu (2008) uses such rank conditions to
identify general nonlinear models with misclassification errors. Assumption
2 in Hu and Shum (2012) use such conditions to identify dynamic models with
unobserved state variables. In both cases, the rank conditions are
introduced to ensure linear independence between component distributions in
a mixture model. In our case, this rank conditions essentially rule out
pathologies where the choices probabilities conditional on unobserved states
are linear combinations of each other. In Section 3.6, we provide numerical
examples that satisfy these rank conditions in MPE. Our calculation also
verifies that Assumption \ref{assum:FR} hold generically in the sense that
continuous variation in the model parameters almost surely implied
equilibrium choice probabilities that satisfy the rank conditions.

This condition is easier to hold when the number of players or the number of choices are large relative to the cardinality of the support of unobserved heterogeneity and equilibria. For example, with $M=2$, $K=4$, and players faces 3 choices. Then we'll need two players in sets $A$ and $B$ each. In the next subsections, we examine situations when the set of players and/or alternatives is not sufficiently large.

Under Assumption \ref{assum:FR}, $\mathbf{P}$ is a finite mixture with each
component being a rank-one matrix $q_{A,j}q_{B,j}^{\prime }$ and mixing
weights given by $w_{j}$. Assumption \ref{assum:FR} also implies that it is
not possible to decompose the $\delta _{A}$-by-$\delta _{B}$ matrix $\mathbf{%
P}$ into other observationally equivalent forms of finite mixtures with
fewer components.\footnote{%
To see this, suppose one can write an alternative decomposition $\mathbf{P}=%
\hat{Q}_{1}\hat{\Lambda}\hat{Q}_{2}^{\prime }$, where $\hat{\Lambda}$ is a
diagonal matrix with dimension $\tilde{J}<MK$. This would imply the rank of $%
\mathbf{P}$ is strictly less than $MK$. On the other hand, Assumption \ref%
{assum:FR} and (\ref{decomp1}) imply that $\mathbf{P}$ must have full rank $%
J=MK$. Contradiction.} It then follows that $MK$ can be generically
identified as the rank of $\mathbf{P}$ under maintained conditions. 

Next, construct a $\#(\mathcal{I})$-vector of actions across two periods%
\begin{equation*}
\tilde{D}_{t}\equiv (D_{A,t-1},D_{B,t+1})\text{.}
\end{equation*}%
By construction, the support of $\tilde{D}_{t}$ and $D_{t}$ are both $%
\mathcal{D}^{\#\mathcal{I}}$ with cardinality $\delta _{A}\delta _{B}=(\#%
\mathcal{D})^{\#\mathcal{I}}$. With slight abuse of notation, let $\lambda
_{m}(\xi ,\xi ^{\prime })\equiv \Pr \{\xi _{t}=\xi ,\xi _{t+1}=\xi ^{\prime
}|m\}$ denote the joint probability mass for time-varying and serially
correlated DUH $\xi _{t}$ and $\xi _{t+1}$ implied in a single MPE
indexed by $m\in \mathcal{M}$; and let $P_{m}(\cdot |\xi ,\xi ^{\prime })$
denote the CCPs given $\xi _{t}=\xi $ and $\xi _{t+1}=\xi ^{\prime }$
implied in that single MPE. Similarly, let $\{\tilde{\omega}%
_{j}\}_{j=1,..,J^{\prime }}$ denote elements of the joint support of the
discrete vector $(m,\xi _{t},\xi _{t+1})$. By construction $J^{\prime
}\equiv MK^{2}$. The joint probability mass of $(\tilde{D}%
_{t},D_{t})$ is 
\begin{equation*}
P(\tilde{D}_{t},D_{t})=\sum\nolimits_{m,\xi ,\xi ^{\prime }}\varphi
(m)\lambda _{m}(\xi ,\xi ^{\prime })P_{m}(\tilde{D}_{t},D_{t}|\xi ,\xi
^{\prime }),
\end{equation*}%
where 
\begin{eqnarray}
&&P_{m}(\tilde{D}_{t},D_{t}|\xi ,\xi ^{\prime
})=P_{m}(D_{B,t+1},D_{A,t-1},D_{t}|\xi ,\xi ^{\prime })  \notag \\
&=&P_{m}(D_{B,t+1}|D_{t},D_{A,t-1},\xi ,\xi ^{\prime
})P_{m}(D_{t}|D_{A,t-1},\xi ,\xi ^{\prime })P_{m}(D_{A,t-1}|\xi ,\xi
^{\prime })  \notag \\
&=&P_{m}(D_{B,t+1}|\xi ^{\prime })P_{m}(D_{t}|\xi ^{\prime },\xi
)P_{m}(D_{A,t-1}|\xi )\text{.}  \label{PDt}
\end{eqnarray}%
The last equality above uses several implications of the assumptions
maintained above: (i)$\ D_{i,t}$ is determined by $(\xi _{t},\epsilon
_{i,t}) $ in each single MPE; (ii) $\epsilon _{i,t}$ is independent of past
histories of $(\epsilon _{s},\xi _{s})_{s\leq t-1}$ once conditional on $\xi
_{t}$, and (iii) $\xi _{t+1}$ is independent of past $(\xi _{s},\epsilon
_{s})_{s\leq t-1}$ once conditional on $\xi _{t}$.

Let $\mathbf{\tilde{P}}$ denote a $\delta _{A}\delta _{B}$-by-$\delta
_{A}\delta _{B}$ matrix that summarizes the joint probability mass of $(%
\tilde{D}_{t},D_{t})$, with rows indexing the realization of $\tilde{D}_{t}$%
\ and columns indexing the realization of $D_{t}$. 
Let $\tilde{\theta}_{j}\equiv \Pr \{(m,\xi _{t},\xi _{t+1})=%
\tilde{\omega}_{j}\}$.
In matrix notation, 
\begin{equation}
\underset{\tilde{Q}}{\mathbf{\tilde{P}}\equiv \underbrace{\left( 
\begin{array}{ccc}
\tilde{q}_{1} & \cdots & \tilde{q}_{J^{\prime }}%
\end{array}%
\right) }}\underset{\tilde{\Lambda}}{\underbrace{\left( 
\begin{array}{ccc}
\tilde{\theta}_{1} & 0 & 0 \\ 
0 & \ddots & 0 \\ 
0 & 0 & \tilde{\theta}_{J^{\prime }}%
\end{array}%
\right) }}\underset{Q^{\prime }}{\underbrace{\left( 
\begin{array}{c}
q_{1}^{\prime } \\ 
\vdots \\ 
q_{J^{\prime }}^{\prime }%
\end{array}%
\right) }=}\sum\nolimits_{j}\tilde{\theta}_{j}\tilde{q}_{j}q_{j}^{\prime }\text{,%
}  \label{decomp2}
\end{equation}%
where $\tilde{q}_{j}$ is a $\delta _{A}\delta _{B}$-by-$1$ column vector 
\begin{equation*}
\{P_{m}(D_{A,t-1}=d_{A}|\xi _{t}=\xi )P_{m}(D_{B,t+1}=d_{B}|\xi _{t+1}=\xi
^{\prime }):(d_{A},d_{B})\in \mathcal{D}^{\#\mathcal{I}}\}
\end{equation*}%
with $(m,\xi ,\xi ^{\prime })=\tilde{\omega}_{j}$; and $q_{j}$ is a $\delta
_{A}\delta _{B}$-by-$1$ column vector 
\begin{equation*}
\{P_{m}(D_{t}=d|\xi _{t}=\xi ,\xi _{t+1}=\xi ^{\prime }):d\in \mathcal{D}%
^{\#\mathcal{I}}\}
\end{equation*}%
with $(m,\xi ,\xi ^{\prime })=\tilde{\omega}_{j}$.

\begin{Assumption}
\label{assum:FR2} Both $\tilde{Q}$ and $Q$ have full rank $MK^{2}$.
\end{Assumption}

Similar to Assumption \ref{assum:FR}, this is a restriction on the choice
probabilities conditional on DUH $\xi $ and equilibrium index $m$.
Also note that by construction, $\min \{\delta _{A},\delta _{B}\}>MK$
implies that $\delta _{A}\delta _{B}>MK^{2}$. It follows from Kasahara and
Shimotsu (2016) that $MK^{2}$ is generically identified as the rank of the $%
\delta _{A}\delta _{B}$-by-$\delta _{A}\delta _{B}$ square matrix $\mathbf{\ 
\tilde{P}}$ under maintained conditions. Consequently, both $M$ and $K$ are
identified from knowledge of $MK=rank(\mathbf{P})$ and $MK^{2}=rank(\mathbf{%
\ \tilde{P}})$. The cardinality of equilibria in the DGP is revealed as $%
\left[ rank(\mathbf{P})\right] ^{2}/rank(\mathbf{\tilde{P}})$.

\begin{proposition}
\label{pn:dynamicTest} Suppose Assumptions \ref{assum:fix_mpe}, \ref%
{assum:FR} and \ref{assum:FR2}\ hold in the dynamic games with private
information above. Then the number of MPE in the data-generating process is $%
\left[ rank(\mathbf{P})\right] ^{2}/rank(\mathbf{\tilde{P}})$.
\end{proposition}

For this method to work, it is essential that we choose to partition the history and focus on the joint probability mass of  $\tilde{D}_t$ and $D_t$ in $\mathbf{\tilde{P}}$. 
If we were to partition the history differently by defining $\mathbf{\tilde{P}}$ as the joint probability mass of $D_t$ and $D_{t+1}$, then its finite mixture representation would only consist of $MK$ components.

The basic idea in this subsection can be extended to accommodate the
transition of observable states in addition to DUH. To do so,
condition the identification method on realization of $X_{t+1}$ and $X_{t}$.
The mixing weights in $\mathbf{P}$\ and $\mathbf{\tilde{P}}$ become $\Pr
\{\xi _{t}=\xi |X_{t}=x,m\}$ and $\Pr \{\xi _{t+1}=\xi ^{\prime },\xi
_{t}=\xi |X_{t+1}=x^{\prime },X_{t}=x,m\}$ respectively. The derivation in (%
\ref{PD}) and (\ref{PDt}) carries through after conditioning on
\textquotedblleft $X_{t}=x$\textquotedblright\ and \textquotedblleft $%
X_{t}=x $ and $X_{t+1}=x^{\prime }$\textquotedblright\ respectively.

We have thus formulated the test for multiple MPE in the data-generating
process as a question of inferring the rank of a consistently estimable
matrix. \cite{kleibergenPaap2006} proposed a rank test which uses the
singular value decomposition and has a pivotal limiting distribution under
the null. \cite{kasaharaShimotsu2014} used this test statistic to construct
a consistent estimator for the number of components in a finite mixture
model, following a sequential testing algorithm in \cite{robinSmith2000}.
The estimator in \cite{kasaharaShimotsu2014} can be used for
our purpose of testing multiple MPE. The asymptotic properties of the rank
test and the estimator are presented in \cite{kleibergenPaap2006} and
\cite{kasaharaShimotsu2014}. Both papers documented evidence of the test
and estimator's finite sample performance via various simulation exercises.

\subsection{Extension: binary decisions with many players}

\label{sec:symMPE}

The method in Section \ref{sec:MPE} requires the support of choice profiles
be larger than the joint support of DUH and equilibria (that is, %
$(\#\mathcal{D})^{\#\mathcal{I}_{A}}$, $(\#\mathcal{D})^{\#\mathcal{I}%
_{B}}>MK$). If choices are binary $\mathcal{D}\equiv \{0,1\}$, one can
extend the logic to test multiple equilibria provided there are sufficiently
many players in a game.

To fix ideas, let's first consider a model with no $X_{t}$ as before.
Suppose the set of players in each dynamic game is partitioned into two
types labelled by $1$ and $2$. Players with the same type are ex ante
identical in that they share the same ex post preference and have
idiosyncratic shocks drawn independently from the same distribution. Let $%
n_{1}$ and $n_{2}$ denote the number of type-1 and type-2 players so that $\#%
\mathcal{I}=n=n_{1}+n_{2}$. A type-symmetric MPE is characaterized by $%
\sigma _{\tau }^{\ast }(\xi ,\varepsilon _{i})\rightarrow \mathcal{D}$ for
types $\tau =1,2$.

We can modify the method in Section \ref{sec:MPE}\ for this setting. Let $%
m_{\tau t}$ denote the number of type-$\tau $ players choosing action $1$ in
period $t$. Then we can construct a $(n_{1}+1)$ -by-$(n_{2}+1)$ matrix of
the joint probability mass of $(m_{1t},m_{2t})$, denoted $\mathbf{S}$, with
the rows and columns indexed by the possible values of $m_{1t}$ and $m_{2t}$%
, and $(i,j)$-th element being probability for $m_{1t}=i+1$ and $m_{2t}=j+1$%
. Likewise we can construct a $\tilde{n}$-by-$\tilde{n}$ matrix for the
joint probability mass of $(m_{1t},m_{2t})$ and $(m_{1,t-1},m_{2,t+1})$,
denoted $\mathbf{\tilde{S}}$, where $\tilde{n}\equiv (n_{1}+1)(n_{2}+1)$.
Applying the law of total probability and exploiting the conditional
independence assumptions, we can decompose $\mathbf{S}$ and $\mathbf{\tilde{S%
}}$ both into products of three matrices in forms similar to (\ref{decomp1})
and (\ref{decomp2}). If $n_{1},n_{2}>MK-1$ and proper rank conditions hold,
then the ranks of $\mathbf{S}$ and $\mathbf{\tilde{S}}$ are equal to the
dimensions of diagonal matrices in the products, or $MK$ and $MK^{2}$
respectively. Thus the cardinality of equilibria in the DGP is identified.

\subsection{Extension: binary decisions with two players\label{sec:2p2d}}

The remaining challenge is to test for multiple MPE when there are only two
players and binary decisions $\#\mathcal{D}=2$ and $\#\mathcal{I}=2$. In
this case, we need to use five consecutive time periods. Let $W_{t}$ denote
a discretization of $(D_{t},X_{t})$ defined by partitioning the support of $%
X_{t}$ into intervals.

Consider the following joint probability mass function for a fixed
realization of the second-period choices and states $\bar{w}_{2}$:%
\begin{eqnarray*}
P(w_{3},\bar{w}_{2},w_{1}) &=&\sum\limits_{m,\xi _{2}}\varphi (m)\lambda (\xi
_{2})P_{m}(w_{3},\bar{w}_{2},w_{1}|\xi _{2}) \\
&=&\sum\limits_{m,\xi _{2}}\varphi (m)\lambda (\xi _{2})\frac{P_{m}(w_{3}|\bar{%
w }_{2},\xi _{2},w_{1})P_{m}(\bar{w}_{2},\xi _{2},w_{1})}{\lambda (\xi _{2})}
\\
&=&\sum\limits_{m,\xi _{2}}\varphi (m)P_{m}(w_{3}|\bar{w}_{2},\xi _{2})P_{m}( 
\bar{w}_{2},\xi _{2},w_{1})\text{,}
\end{eqnarray*}%
where $\varphi (\cdot )$ is the equilibrium selection; $\lambda (\cdot )$ is the
marginal probability mass function of $\xi _{2}$ under MPE $m\in \mathcal{M}$%
. The last equality above uses conditional independence conditions
maintained on the state and DUH transitions. Let $\delta _{(1)}$ and $\delta
_{(3)}$ denote the cardinality of the marginal support of $W_{1}$ and $W_{3}$
respectively.

Let $\mathbf{P}_{w_{3},\bar{w}_{2},w_{1}}$ denote a $\delta _{(3)}$-by-$%
\delta _{(1)}$ matrix that summarizes the joint probability mass function of 
$(W_{3},W_{2},W_{1})$ when the realization of $W_{2}$ is fixed at $\bar{w}_{2}$.
Let the rows in $\mathbf{P}_{w_{3},\bar{w}_{2},w_{1}}$ be indexed by
elements on the support of $W_{3}$ and the columns by elements on the
support of $W_{1}$. In matrix notation,%
\begin{equation*}
\mathbf{P}_{w_{3},\bar{w}_{2},w_{1}}\equiv \left( \Phi _{\bar{w}_{2}}\right)
\Pi \left( \Psi _{\bar{w}_{2}}\right) ^{\prime }
\end{equation*}%
where $\Pi $ is a $J$-by-$J$ diagonal matrix with non-zero diagonal entries
being $\{\varphi (m):m\in \mathcal{M}\}$ (each $\varphi (m)$ is repeated $K$ times
on the diagonal); $\Phi _{\bar{w}_{2}}$ is a $\delta _{(3)}$-by-$J$ matrix
with its $(i,j)$-th component being $P_{m}(w_{3}|\bar{w}_{2},\xi _{2})$
where $(m,\xi _{2})=\omega _{j}$ and $w_{3}$ is the $i$-th element in
support of $W_{3}$; and likewise $\Psi _{\bar{w}_{2}}$ is a $\delta _{(1)}$%
-by-$J$ matrix with its $(i,j)$-th component being $P_{m}(\bar{w}_{2},\xi
_{2},w_{1})$ where $(m,\xi _{2})=\omega _{j}$ and $w_{1}$ is the $i$-th
element in support of $W_{1}$. With $\delta _{(3)},\delta _{(1)}>J=MK$ and $%
\Phi _{\bar{w}_{2}}$, $\Psi _{\bar{w}_{2}}$ both having full rank, we
identify $J=MK$ as the rank of $\mathbf{P}_{w_{3},\bar{w}_{2},w_{1}}$.

Next, consider the following joint probability mass function for a fixed
pair of realization in the second and fourth period $(\bar{w}_{4},\bar{w}%
_{2})$:%
\begin{equation*}
P(w_{5},\bar{w}_{4},w_{3},\bar{w}_{2},w_{1})=\sum\limits_{m,\xi _{4,}\xi
_{2}}\varphi (m)\lambda _{m}(\xi _{4},\xi _{2})P_{m}(w_{5},\bar{w}_{4},w_{3},%
\bar{w}_{2},w_{1}|\xi _{4},\xi _{2})\text{,}
\end{equation*}%
where $\lambda _{m}(\cdot
,\cdot )$ is the joint distribution of $(\xi _{2},\xi _{4})$ under MPE $m\in 
\mathcal{M}$; and for each MPE $m$ and conditioning on $(\xi _{4},\xi _{2})$%
, the joint probability mass of the history is%
\begin{eqnarray}
&&P_{m}(w_{5},\bar{w}_{4},w_{3},\bar{w}_{2},w_{1}|\xi _{4},\xi _{2})
\label{T5} \\
&=&P_{m}(w_{5}|\bar{w}_{4},\xi _{4},w_{3},\bar{w}_{2},\xi _{2},w_{1})P_{m}(%
\bar{w}_{4},w_{3},\bar{w}_{2}|w_{1},\xi _{4},\xi _{2})P_{m}(w_{1}|\xi
_{4},\xi _{2})\text{.}  \notag
\end{eqnarray}%
By our maintained assumptions,%
\begin{equation*}
P_{m}(w_{5}|\bar{w}_{4},\xi _{4},w_{3},\bar{w}_{2},\xi
_{2},w_{1})=P_{m}(w_{5}|\bar{w}_{4},\xi _{4})\text{,}
\end{equation*}%
\begin{eqnarray*}
P_{m}(\bar{w}_{4},w_{3},\bar{w}_{2}|w_{1},\xi _{4},\xi _{2}) 
&=&\frac{P_{m}(\bar{w}_{4},\xi _{4}|w_{3},\bar{w}_{2},\xi
_{2},w_{1})P_{m}(w_{3}|\bar{w}_{2},\xi _{2},w_{1})P_{m}(\bar{w}_{2},\xi
_{2}|w_{1})}{P_{m}(\xi _{4},\xi _{2}|w_{1})} \\
&=&\frac{P_{m}(\bar{w}_{4},\xi _{4}|w_{3},\bar{w}_{2},\xi _{2})P_{m}(w_{3}| 
\bar{w}_{2},\xi _{2})P_{m}(\bar{w}_{2},\xi _{2}|w_{1})}{P_{m}(\xi _{4},\xi
_{2}|w_{1})}\text{,}
\end{eqnarray*}%
and%
\begin{equation*}
P_{m}(w_{1}|\xi _{4},\xi _{2})=\frac{P_{m}(\xi _{4},\xi _{2}|w_{1})P(w_{1})}{
\lambda _{m}(\xi _{2},\xi _{4})}\text{.}
\end{equation*}%
Substituting this into (\ref{T5}), we get 
\begin{eqnarray*}
&&P(w_{5},\bar{w}_{4},w_{3},\bar{w}_{2},w_{1}) \\
&=&\sum\limits_{m,\xi _{4,}\xi _{2}}\varphi (m)P_{m}(w_{5}|\bar{w}_{4},\xi
_{4})P_{m}(\bar{w}_{4},\xi _{4},w_{3}|\bar{w}_{2},\xi _{2})P_{m}(\bar{w}
_{2},\xi _{2},w_{1})\text{.}
\end{eqnarray*}%
Let $\delta _{(5,1)}$ denote the cardinality of the support of $%
(W_{5},W_{1}) $.

Let $\mathbf{\tilde{P}}_{w_{5},\bar{w}_{4},w_{3},\bar{w}_{2},w_{1}}$ denote
a $\delta _{(5,1)}$-by-$\delta _{(3)}$ matrix that summarizes the joint
probability mass function of $(W_{5},W_{4},W_{3},W_{2},W_{1})$ when $%
(W_{4},W_{2})$ are fixed at $(\bar{w}_{4},\bar{w}_{2})$. Let the rows in $%
\mathbf{\tilde{P}}_{w_{5},\bar{w}_{4},w_{3},\bar{w}_{2},w_{1}}$ be indexed
by elements on the joint support of $(W_{5},W_{1})$ and the columns by
elements on the marginal support of $W_{3}$. In matrix notation, 
\begin{equation*}
\mathbf{\tilde{P}}_{w_{5},\bar{w}_{4},w_{3},\bar{w}_{2},w_{1}}\equiv \left( 
\tilde{\Phi}_{\bar{w}_{4},\bar{w}_{2}}\right) \tilde{\Pi}\left( \tilde{\Psi}%
_{\bar{w}_{4},\bar{w}_{2}}\right) ^{\prime }
\end{equation*}%
where $\tilde{\Pi}$ is a $J^{\prime }$-by-$J^{\prime }$ diagonal matrix with
non-zero diagonal entries being $\{\varphi (m):m\in \mathcal{M}\}$ (each $\varphi
(m) $ is repeated $K^{2}$ times on the diagonal); $\tilde{\Phi}_{\bar{w}%
_{4}, \bar{w}_{2}}$ is a $\delta _{(5,1)}$-by-$J^{\prime }$ matrix with its $%
(i,j)$-th component being $P_{m}(w_{5}|\bar{w}_{4},\xi _{4})P_{m}(\bar{w}%
_{2},\xi _{2},w_{1})$ where $(m,\xi _{2},\xi _{4})=\tilde{\omega}_{j}$ and $%
(w_{5},w_{1})$ being the $i$-th element in support of $(W_{5},W_{1})$, and
likewise $\tilde{\Psi}_{\bar{w}_{4},\bar{w}_{2}}$ is a $\delta _{(3)}$-by-$%
J^{\prime }$ matrix with its $(i,j)$-th component being $P_{m}(\bar{w}%
_{4},\xi _{4},w_{3}|\bar{w}_{2},\xi _{2})$ where $(m,\xi _{2},\xi _{4})=%
\tilde{\omega}_{j}$ and $w_{3}$ is the $i$-th element in support of $W_{3}$.
With $\delta _{(5,1)},\delta _{(3)}>J^{\prime }$ and $\left( \tilde{\Phi}_{ 
\bar{w}_{4},\bar{w}_{2}}\right) $ and $\left( \tilde{\Psi}_{\bar{w}_{4},\bar{%
		w}_{2}}\right) $ being full rank, we identify $J^{\prime }=MK^{2}$ as the
rank of $\mathbf{\tilde{P}}_{w_{5},\bar{w}_{4},w_{3},\bar{w}_{2},w_{1}}$.

For this method to work, it is crucial that we partition the 
history into $(W_{1},W_{5})$ and $W_{3}$ while indexing the rows and columns 
$\mathbf{\tilde{P}}_{w_{5},\bar{w}_{4},w_{3},\bar{w}_{2},w_{1}}$. To see
how, suppose we had defined $\mathbf{\tilde{P}}_{w_{5},\bar{w}_{4},w_{3},%
\bar{w}_{2},w_{1}}$ differently as a $\delta _{(5,3)}$-by-$\delta _{(1)}$
matrix with rows indexing $(w_{5},w_{3})$ and columns index $w_{1}$. Then it
would have rank $J=MK$ as opposed to $J^{\prime }=MK^{2}$.

\subsection{Related literature}

\cite{otsu_et_al2016} tested the null hypothesis that
choices observed in a finite number of games/markets $s=1,2,...,S<\infty $
are generated by the same MPE. With the number of time periods $T$ in each
game approaching infinity, they proposed an asymptotic test based on the
estimation and comparison of conditional choice probabilities across the
finite number of games. In comparison, we focus on a different scenario
where a sample is large in the number of independent dynamic games/markets
(that is, asymptotics is defined as $S\rightarrow \infty $) while the number
of time periods $T$ observed in each game is small and finite (say $T\leq 3$%
).

\cite{huShum2012} used an eigendecomposition method to fully identify
general dynamic models with unobserved heterogeneity. \cite{luo_et_al2019} extended this method to deal with multiple equilibria in dynamic
games by treating equilibrium selection as a component in discrete unobserved
heterogeneity on the game level. The key in their paper was to define a set
of conditioning events so that the conditional distribution of observed
outcomes admitted a finite-mixture (or eigenvalue decomposition)
representation with $MK$ components, with $K$ being cardinality of
unobserved heterogeneity and $M$ being cardinality of equilibria in the DGP.
\cite{luo_et_al2019} then proposed to test multiple equilibria by
checking whether player payoffs identified from the component CCPs in the
finite mixture (which themselves need to be recovered from outcome
distributions via eigenvalue decomposition) are distinctively different.
This idea was built on the same insight that \cite{aguirregabiriaMira2019}
introduced for static games: if there are multiple equilibria in the DGP,
then player payoffs backed out from some of the component CCPs in the
mixture would be identical.\footnote{%
\cite{luo_et_al2019} described the test in sequential steps: (1)
recover component CCPs given DUH and equilibrium
selection using eigendecomposition, (2) estimate payoffs from component
CCPs, and (3) comparing payoff estimates calculated from different CCPs. The
costs for nonparametrically implementation of these steps are high. See our discussion about analogous costs for static games in the second paragraph of Section 2.4.
\cite{luo_et_al2019} did not provide any test statistic or asymptotic
theory for implementing this test in their paper.}

In contrast, the insight in our new approach is different. 
We separate the task of detecting multiple equilibria from identifying the full model. By partitioning and
pairing outcome histories with different lengths, we can represent their
joint distribution as finite mixtures whose cardinality of components 
\textit{vary with the length of history in addition to }$M$\textit{\ and }$K$%
. This is a new source of variation that has not been exploited in the
literature to the best of our knowledge. More importantly, it allows us to
conduct simple inference for equilibrium cardinality through standard rank
test such as \cite{kleibergenPaap2006}, thus dispensing with the sequential
steps discussed in \cite{luo_et_al2019}.\footnote{%
In our method, we only need to recover the ranks of the joint distribution
of outcome histories. As a result, we do not need to actually perform an
eigendecomposition. This means we do not need the finite mixture to
take a form of $Q\Lambda Q^{-1}$. Expressing the finite mixture in a much
less restrictive form $Q\Lambda \tilde{Q}$ is totally fine for our purpose.}

\subsection{Numerical examples}

\noindent \textbf{Example 1}. (Multiple players with binary choices.)
Consider a dynamic game of imperfect information between $n$ players. Each
player $i$ belongs to one of two types (labeled as 1 or 2) and makes a
binary decision $d_{i,t}$ in each period $t$. Same-type players have
identical ex post payoffs and have idiosyncratic shocks drawn independently
from the same distribution. Let $n_{1}$ and $n_{2}$ denote the number of
type-1 and type-2 players respectively.

A time-varying state variable $\xi _{t}\in \{0,1\}$ is commonly observed by
all players in each period $t$. In addition, each player observes a vector
of private signals in each period: $\varepsilon _{i,t}\equiv (\varepsilon
_{i,t,s})_{s=0,1}\in \mathbb{R}^{2}$. Let $\varepsilon _{t}\equiv
(\varepsilon _{i,t})_{i\leq n}$. The state-and-signal transition is such
that $\Pr \{\xi _{t+1},\varepsilon _{t+1}|\xi _{t},\varepsilon
_{t},d_{t}\}=\Pr \{\varepsilon _{t+1}\}\Pr \{\xi _{t+1}|\xi _{t},d_{t}\}$,
where $d_{t}\equiv (d_{i,t})_{i\leq n}$ and $\varepsilon _{i,t}$ are i.i.d.
across $i\leq n$ and independent form $\xi _{t}$. For each $i$ and $t$, $%
\varepsilon _{i,t}$ is bivariate normal with zero mean and an identity
covariance matrix. The law of transition for $\xi_t $ and the ex post payoffs
depend on $d_{t}$ only through the number of type-1 and type-2 players who
chooses $1$ in period $t$, denoted $m_{t}\equiv (m_{1t},m_{2t})$. The
transition $\Pr \{\xi _{t+1}=1|\xi _{t}=\xi ,m_{t}\}=p_{\xi ,1}$ if $%
m_{1t}/n_{1}\geq m_{2t}/n_{2}$; and $p_{\xi ,0}$ otherwise. The ex post
payoff of a type-$\tau $ player $i$ choosing $d_{i,t}=d_{i}\in \{0,1\}$ is
given by 
\begin{equation*}
\pi _{\tau }(d_{i},d_{-i,t},\xi _{t})+\varepsilon _{i,t,d_{i}}\text{,}
\end{equation*}%
where $\pi _{\tau }(d_{i},d_{-i,t},\xi )=c_{\tau }(d_{i},\xi )+\delta _{\tau
1}(d_{i})\left( m_{1t}/n_{1}\right) +\delta _{\tau 2}(d_{i})\left(
m_{2t}/n_{2}\right) $ for type $\tau =1,2$.

Consider a game with ten players with $n_{1}=n_{2}=5$ and a discount factor $%
\beta =0.75$. The other game parameters are specified as follows:%
\begin{equation*}
\begin{tabular}{|lllll|}
\hline
& \multicolumn{2}{c}{$c_{\tau =1}(d_{i},\xi )$} & \multicolumn{2}{c|}{$%
c_{\tau =2}(d_{i},\xi )$} \\ \hline
& $d_{i}=0$ & $d_{i}=1$ & $d_{i}=0$ & $d_{i}=1$ \\ \hline\hline
$\xi =0$ & $2.6$ & $-0.8$ & $1.4$ & $-1.8$ \\ 
$\xi =1$ & $1.4$ & $-1.8$ & $2.6$ & $-0.8$ \\ \hline
\end{tabular}%
\text{ ; }%
\begin{tabular}{|lllll|}
\hline
& \multicolumn{2}{c}{$\delta _{\tau =1,\tilde{\tau}}(d_{i})$} & 
\multicolumn{2}{c|}{$\delta _{\tau =2,\tilde{\tau}}(d_{i})$} \\ \hline
& $d_{i}=0$ & $d_{i}=1$ & $d_{i}=0$ & $d_{i}=1$ \\ \hline\hline
$\tilde{\tau}=1$ & $2.5$ & $4$ & $2.5$ & $6.5$ \\ 
$\tilde{\tau}=2$ & $2.5$ & $6.5$ & $2.5$ & $4$ \\ \hline
\end{tabular}%
\text{;}
\end{equation*}%
and 
\begin{equation*}
p_{0,0}=0.375\text{; }p_{0,1}=0.725\text{; }p_{1,0}=0.675\text{; }%
p_{1,1}=0.125\text{.}
\end{equation*}%
These parameter values reflect certain economic interpretation: First, for
different types of players, the marginal impacts of unobserved states $\xi $
on $c_{\tau }(d_{i},\xi )$ move in different directions. Second, when $%
d_{i}=1$, $\delta _{\tau ,\tilde{\tau}}$ reflects higher complementarity
from other types of players choosing the same action. Third, the transition
of states depends on the choice profiles summarized by ($m_{1t},m_{2,t}$),
which has a substantial impact on the likelihood of transition to higher
states. We also allow the sign of such impact depend on the current state as
well.

We solve for conditional choice probabilities in type-symmetric MPE. The
game admits two MPE's that lead to different vectors of CCPs:%
\begin{eqnarray*}
&&\text{Conditional Choice Probabilities (CCPs) in MPE} \\
&&\text{ \ \ \ \ \ \ \ }%
\begin{tabular}{lllll}
\hline
& \multicolumn{2}{c}{$Eq\#1$} & \multicolumn{2}{c}{$Eq\#2$} \\ 
& $\tau =1$ & \multicolumn{1}{l|}{$\tau =2$} & $\tau =1$ & $\tau =2$ \\ 
\hline\hline
$\xi =0$ & $0.473$ & \multicolumn{1}{l|}{$0.533$} & $0.041$ & $0.022$ \\ 
$\xi =1$ & $0.375$ & \multicolumn{1}{l|}{$0.309$} & $0.304$ & $0.334$ \\ 
\hline
\end{tabular}%
\end{eqnarray*}

To illustrate the method in Section \ref{sec:MPE}, suppose both equilibria
are selected with positive probability in the data generating process. With $%
A$ being the set of type-1 players and $B$ being the set of type-2 players,
the joint probability mass $\mathbf{P}$ in (\ref{decomp1}) is a $2^{5}$-by-$%
2^{5}$ with rank $MK=4$. Furthermore, the other joint probability mass $%
\mathbf{\tilde{P}}$ in (\ref{decomp2}) is $2^{10}$-by-$2^{10}$ matrix with
rank $MK^{2}=8$. It is verified that this holds because the rank conditions
in Assumptions \ref{assum:FR} and \ref{assum:FR2}\ are satisfied.

It is worth mentioning that this idea for testing multiple MPE remains valid
when applied to lower-dimension matrices that aggregate over the rows and
columns in $\mathbf{P}$ and $\mathbf{\tilde{P}}$. For instance, one may well
replace $A$ and $B$ by arbitrarily picked subsets of type-1 or type-2
players, say, three from each type. Then one can construct two joint prob
mass matrices similar to $\mathbf{P}$ and $\mathbf{\tilde{P}}$ for these
players only. These matrices will be $2^{3}$-by-$2^{3}$ and $2^{6}$-by-$%
2^{6} $ in dimensions, and have ranks $4$ and $8$ respectively. 

We also illustrate how to implement the alternative method in Section \ref%
{sec:symMPE}. To do so, construct a $6$-by-$6$ matrix of joint probability
masses $\mathbf{S}$ for $(m_{1t},m_{2t})$, i.e., the number of type-1 and
type-2 players choosing $1$. Likewise construct a $6^{2}$-by-$6^{2}$ matrix
of joint probability masses $\mathbf{\tilde{S}}$ for $m_{t}\equiv
(m_{1t},m_{2t})$ and $\tilde{m}_{t}\equiv (m_{1,t-1},m_{2,t+1})$. Both
matrices admit a diagonalized forms similar to (\ref{decomp1}) and (\ref%
{decomp2}), with the diagonal matrices in the middle of the decompositions
being the same as $\Lambda $ and $\tilde{\Lambda}$ in (\ref{decomp1}) and (%
\ref{decomp2}). It is verified that the outer matrices in both
decomponsitions, a.k.a. component mass functions conditional on equilibrium
selection and unobserved states, satisfy the appropriate rank conditions.
Therefore the probability mass matrices $\mathbf{S}$ and $\mathbf{\tilde{S}}$
have ranks $MK=4$ and $MK^{2}=8$ respectively. As noted, one
can also perform the test by replacing $\mathbf{\tilde{S}}$ with
lower-dimension transformation, such as $\tilde{n}$-by-$\tilde{n}$ summary
matrix (with $8<\tilde{n}<36$) formed by some linear combinations of rows
and columns in $\mathbf{\tilde{S}}$.

Given the equilibrium choice probabilities reported in Example 1A above, the
numeric values of the probability mass matrices are%
\begin{equation*}
\mathbf{S}=\left( 
\begin{array}{cccccc}
0.1588 & 0.0285 & 0.0152 & 0.0113 & 0.0053 & 0.0011 \\ 
0.0449 & 0.0390 & 0.0510 & 0.0443 & 0.0221 & 0.0047 \\ 
0.0176 & 0.0485 & 0.0770 & 0.0728 & 0.0381 & 0.0084 \\ 
0.0098 & 0.0351 & 0.0616 & 0.0618 & 0.0333 & 0.0074 \\ 
0.0034 & 0.0135 & 0.0256 & 0.0268 & 0.0147 & 0.0033 \\ 
0.0005 & 0.0022 & 0.0044 & 0.0047 & 0.0026 & 0.0006%
\end{array}%
\right) ;
\end{equation*}%
and%
\begin{equation*}
\mathbf{\tilde{S}}^{\ast }=\left( 
\begin{array}{ccccccccc}
0.0615 & 0.0249 & 0.0154 & 0.0339 & 0.0132 & 0.0105 & 0.0071 & 0.0011 & 
0.0008 \\ 
0.0272 & 0.0153 & 0.0177 & 0.0300 & 0.0108 & 0.0127 & 0.0072 & 0.0013 & 
0.0010 \\ 
0.0249 & 0.0134 & 0.0194 & 0.0332 & 0.0128 & 0.0179 & 0.0110 & 0.0020 & 
0.0017 \\ 
0.0407 & 0.0289 & 0.0443 & 0.0694 & 0.0250 & 0.0335 & 0.0188 & 0.0034 & 
0.0027 \\ 
0.0120 & 0.0116 & 0.0188 & 0.0273 & 0.0092 & 0.0114 & 0.0058 & 0.0010 & 
0.0007 \\ 
0.0147 & 0.0126 & 0.0203 & 0.0340 & 0.0132 & 0.0179 & 0.0111 & 0.0020 & 
0.0016 \\ 
0.0079 & 0.0067 & 0.0111 & 0.0162 & 0.0055 & 0.0074 & 0.0038 & 0.0007 & 
0.0005 \\ 
0.0011 & 0.0012 & 0.0019 & 0.0034 & 0.0014 & 0.0020 & 0.0013 & 0.0002 & 
0.0002 \\ 
0.0008 & 0.0007 & 0.0011 & 0.0020 & 0.0008 & 0.0012 & 0.0007 & 0.0001 & 
0.0001%
\end{array}%
\right)
\end{equation*}%
The rows and columns of $\mathbf{S}$ correspond to the number of type-1 and
type-2 players choosing $1$ respectively (ordered from $m=0,1,...,5$). The $%
9 $-by-$9$ matrix $\mathbf{\tilde{S}}^{\ast }$ is a coarsening of the
original $36$-by-$36$ matrix $\mathbf{\tilde{S}}$. It is constructed by
adding up every four adjacent rows and four adjacent columns in $\mathbf{%
\tilde{S}}$. The rank of $\mathbf{S}$ equals $4=MK$; the rank(s) of $\mathbf{%
\tilde{S}}$ and $\mathbf{\tilde{S}}^{\ast }$ are both equal to $8=MK^{2}$%
.\medskip

\noindent \textbf{Example 2.} (Two players with binary decisions.) Consider
a dynamic game between two forward-looking players $i$ and $j$ who make
dynamic, optimal binary decisions $d_{it},d_{jt}\in \{0,1\}$ in each period $%
t$. Players observe both states that evolve over time: $z_{t}\in \{0,1\}$
and $\xi _{t}\in \{-1,1\}$, with $z_{t}$ reported in the data while $\xi
_{t} $ is not. At time $t$, player $i$ observes a vector of private signals $%
\varepsilon _{it}\equiv (\varepsilon _{it1},\varepsilon _{it0})\in \mathbb{R}%
^{2}$. Let $\varepsilon _{t}\equiv (\varepsilon _{it},\varepsilon _{jt})$,
and $\varepsilon _{it},\varepsilon _{jt},z_{t}$ be independent conditional
on $\xi _{t}$. Let $s_{t}\equiv (z_{t},\xi _{t})$ denote the vector of
states, and its transition satisfies%
\begin{equation*}
\Pr (s_{t+1},\varepsilon _{t+1}|s_{t},\varepsilon _{t},d_{i,t},d_{j,t})=\Pr
(\varepsilon _{t+1}|s_{t+1})\Pr (s_{t+1}|s_{t},d_{i,t},d_{j,t})\text{,}
\end{equation*}%
Conditional on $\xi _{t}$, $\varepsilon _{it}$ and $\varepsilon _{jt}$ are
both bivariate normal with mean $(\xi _{t},0)$ and identity covariance. The
state transition is 
\begin{eqnarray*}
&&\Pr \{z_{t+1}=1,\xi _{t+1}|z_{t},\xi
_{t},d_{i,t},d_{j,t}\}=E(z_{t+1}|z_{t},\xi _{t},d_{it},d_{jt})\Pr (\xi
_{t+1}|\xi _{t})\text{;} \\
&&E\left( z_{t+1}|z_{t},\xi _{t},d_{it},d_{jt}\right) =\delta (z_{t},\xi
_{t},d_{it}+d_{jt})\text{.}
\end{eqnarray*}%
The evolution of 
$\xi _{t}$ is serially correlated, but does not depend on observed state or
choices $d_{it},d_{jt}$. In this model, dynamics exist because $%
d_{it},d_{jt} $ affect the transition of observed states $z_{t}$. Ex post
payoff of $i$ from $d_{i}=\tau $ in time $t$ is $\pi (\tau
,d_{j},z)+\varepsilon _{it\tau } $, where $\pi (d_{it},d_{jt},z_{t})=c_{z}/2$
if $d_{it}=d_{jt}$; $\pi (d_{i},d_{j},z)=\alpha c_{z}$ if $d_{it}>d_{jt}$;
and $\pi (d_{it},d_{jt},z_{t})=(1-\alpha )c_{z}$ otherwise.

The specification is chosen to mimic a stylized game of dynamic labor force
participation decisions by a married couple $i$ and $j$ in a household. In
this case, $z_{t}\in \{0,1\}$ reflects the level of household savings at
time $t$ (with $z_{t}=1$ denoting high savings); and $\xi _{t}$ is a series
of time-varying household shocks. Dynamics exist in because the transition
of savings depends on both husband and wife's decision to participate in the
labor force. On the other hand, the transition of shocks do not depend on
the couple's decisions. In each period, $c_{z}$ is the per-period income or
consumption budget, where $c_{1}>c_{0}$ reflect positive association with
savings. The share $\alpha \in (0,1)$ is the share claimed by household
members, which depends on their labor participation decisions.

We solve this game with a discount factor $\beta =0.7$, and parameter values
specified as follows: $\alpha =0.6$, $c_{0}=1.5$, $c_{1}=3$, $\Pr \{\xi
_{t+1}=-1|\xi _{t}=-1\}=0.5$, $\Pr \{\xi _{t+1}=-1|\xi _{t}=1\}=0.3$, and $%
E(z_{t+1}|z_{t},\xi _{t},d_{it},d_{jt})$ specified as%
\begin{equation*}
\begin{tabular}{ccccccc}
\hline
& \multicolumn{2}{c}{$d_{it}+d_{jt}=0$} & \multicolumn{2}{c}{$%
d_{it}+d_{jt}=1 $} & \multicolumn{2}{c}{$d_{it}+d_{jt}=2$} \\ 
& $\xi =-1$ & $\xi =1$ & $\xi =-1$ & $\xi =1$ & $\xi =-1$ & 
\multicolumn{1}{c|}{$\xi =1$} \\ \hline\hline
$z_{t}=0$ & $0.0625$ & $0.1875$ & $0.4375$ & $0.5625$ & $0.6875$ & $0.8125$
\\ 
$z_{t}=1$ & $0.1250$ & $0.25$ & $0.5$ & $0.625$ & $0.75$ & $0.8750$ \\ \hline
\end{tabular}%
\text{,}
\end{equation*}%
which, in the joint labor participation example above, suggests transition
to higher states (savings) is more likely when both players participate.

This specification leads to two symmetric MPE in which the conditional
choice probabilities for participation conditional on common states are:%
\begin{equation*}
\begin{tabular}{ccccc}
\hline
& \multicolumn{2}{c}{$MPE\#1$} & \multicolumn{2}{c}{$MPE\#2$} \\ 
& $\xi =-1$ & $\xi =1$ & $\xi =-1$ & $\xi =1$ \\ \hline\hline
$z_{t}=0$ & $0.272$ & $0.786$ & $0.310$ & $0.826$ \\ 
$z_{t}=1$ & $0.216$ & $0.822$ & $0.347$ & $0.852$ \\ \hline
\end{tabular}%
\text{.}
\end{equation*}%
Plugging in these numbers in the definition of $\mathbf{P}_{w_{3},\bar{w}%
_{2},w_{1}}$ and $\mathbf{\tilde{P}}_{w_{5},\bar{w}_{4},w_{3},\bar{w}%
_{2},w_{1}}$ in Section \ref{sec:2p2d}\ verifies the rank conditions for all
realization of $\bar{w}_{2},\bar{w}_{4}$.

Consider $P_{w_{3},\bar{w}_{2},w_{1}}$. The value of $\bar{w}_{2}\ $is $\bar{%
z}_{2}=1$, $\bar{d}_{i2}=1$ $\bar{d}_{j2}=0$. The dimension of the matrix is 
$8${-by-}$8$, where the rows and columns correspond to a specific action and
observed state in period 3 and period 1 respectively. See below the matrix, $%
P_{w_{3},\bar{w}_{2},w_{1}}$ when the equilibrium selection probability is $%
0.4$. Its rank is equal to $4=MK$. 

\bigskip 

\small
\noindent

\begin{tabular}{c|cccccccc}
 & \multicolumn{8}{c}{$[z_1,d_{i1},d_{j1}]$} \\
$[z_3,d_{i3},d_{j3}]$ & $ [0,0,0]$ & $[0,1,0]$ & $[0,0,1]$ & $[0,1,1]$ & $[1,0,0]$ & $[1,1,0]$ & $[1,0,1]$ & $[1,1,1]$ \\ \hline $[0,0,0]$ & 0.0027 & 0.0103 & 0.0106 & 0.0266 & 0.0038 & 0.0086 & 0.0106 & 0.0261 \\ $[0,1,0]$ & 0.0019 & 0.0077 & 0.0079 & 0.0205 & 0.0027 & 0.0065 & 0.0077 & 0.0200 \\ $[0,0,1]$ & 0.0020 & 0.0079 & 0.0082 & 0.0212 & 0.0028 & 0.0066 & 0.0081 & 0.0208 \\ $[0,1,1]$ & 0.0047 & 0.0190 & 0.0195 & 0.0526 & 0.0066 & 0.0160 & 0.0192 & 0.0515 \\ $[1,0,0]$ & 0.0027 & 0.0107 & 0.0110 & 0.0290 & 0.0038 & 0.0089 & 0.0110 & 0.0284 \\ $[1,1,0]$ & 0.0020 & 0.0082 & 0.0084 & 0.0231 & 0.0028 & 0.0071 & 0.0081 & 0.0225 \\ $[1,0,1]$ & 0.0025 & 0.0099 & 0.0102 & 0.0277 & 0.0035 & 0.0081 & 0.0101 & 0.0271 \\ $[1,1,1]$ & 0.0065 & 0.0266 & 0.0274 & 0.0776 & 0.0090 & 0.0222 & 0.0268 & 0.0758 \\ \hline 
\end{tabular}

\normalsize

\bigskip 

Next, consider $\tilde{P}_{w_{5},\bar{w}_{4},w_{3},\bar{w}_{2},w_{1}}$. The
value of $\bar{w}_{2}\ $and $\bar{w}_{4}$ are $\bar{z}_{2}=1$, $\bar{d}%
_{i2}=1$ $\bar{d}_{j2}=0$ and  $\bar{z}_{4}=1$, $\bar{d}_{i4}=1$ $\bar{d}%
_{j4}=0$ respectively. The original dimension of the matrix is $64${-by-}$8$%
, with the rows correspond to a specific action and observed state in period
1 and period 5 and the columns correspond to a specific action and observed
state in period 3. As noted above, to reduce dimension, we collapsed $\tilde{%
P}_{w_{5},\bar{w}_{4},w_{3},\bar{w}_{2},w_{1}}$ into a $8${-by-}$8$ matrix,
by adding the rows that share the same ${d}_{i5},{z}_{5}${\ and }${d}_{j1}$.
See below the collapsed matrix $\tilde{P}_{w_{5},\bar{w}_{4},w_{3},\bar{w}%
_{2},w_{1}}$, when the equilibrium selection probability is $0.4$, and
hence. Its rank is equal to $8=MK^{2}$. \bigskip

\small
\noindent

\begin{tabular}{c|cccccccc} 
 & \multicolumn{8}{c}{$[z_3,d_{i3},d_{j3}]$} \\
$[d_{j1},z_5,d_{i5}]$ & $ [0,0,0]$ & $[0,1,0]$ & $[0,0,1]$ & $[0,1,1]$ & $[1,0,0]$ & $[1,1,0]$ & $[1,0,1]$ & $[1,1,1]$ \\ \hline $[0,0,0]$ & 0.0006 & 0.0029 & 0.0030 & 0.0106 & 0.0012 & 0.0037 & 0.0041 & 0.0159 \\ $[0,0,1]$ & 0.0009 & 0.0044 & 0.0045 & 0.0167 & 0.0017 & 0.0056 & 0.0062 & 0.0248 \\ $[0,1,0]$ & 0.0007 & 0.0033 & 0.0034 & 0.0126 & 0.0013 & 0.0041 & 0.0047 & 0.0187 \\ $[0,1,1]$ & 0.0012 & 0.0059 & 0.0061 & 0.0233 & 0.0022 & 0.0075 & 0.0082 & 0.0346 \\ $[1,0,0]$ & 0.0018 & 0.0087 & 0.0089 & 0.0324 & 0.0036 & 0.0112 & 0.0128 & 0.0505 \\ $[1,0,1]$ & 0.0026 & 0.0131 & 0.0134 & 0.0506 & 0.0052 & 0.0170 & 0.0191 & 0.0789 \\ $[1,1,0]$ & 0.0020 & 0.0099 & 0.0102 & 0.0383 & 0.0039 & 0.0126 & 0.0146 & 0.0599 \\ $[1,1,1]$ & 0.0034 & 0.0175 & 0.0180 & 0.0706 & 0.0066 & 0.0228 & 0.0255 & 0.1100 \\ \hline 
\end{tabular}

\bigskip 

\normalsize

\section{Concluding Remarks}

This article studies testable implications of multiple equilibria in discrete games with incomplete information where players'
private signals are allowed to be correlated. 
In static games, independence across games whose equilibrium selection is correlated can be used to overcome difficulties in testing for multiple equilibria. 
In dynamic games with serially correlated discrete unobserved heterogeneity, the distribution of an observed history (sequence) of choices and states is a finite mixture over equilibria and unobserved heterogeneity.  
Our testable implication exploits the fact that the number of mixture components is a \textit{known} function of the horizon of the history (length of the sequence), as well as the cardinality of equilibria and unobserved heterogeneity support.  
In both static and dynamic cases, the testable implications are conducive to formal tests using existing statistical tools. 

Further connections between these strategies might allow us to relax the independence across games used in the testable implications of static games if private types are correlated across games through DUH. It is possible to employ the ideas in the dynamic settings we considered, in which the DUH is serially correlated. Exploiting the two jointly might allow for even more powerful detection of multiplicity, but would also require new testing protocols for which further research is needed.

\bibliographystyle{ecta}
\bibliography{dePaula_Tang_2020}

\end{document}